\pdfoutput=1
\documentclass[journal]{IEEEtran}
\usepackage[utf8]{inputenc}
\usepackage{amsmath}
\usepackage{amsfonts}
\usepackage{bbding}
\usepackage{amssymb}
\usepackage{array}
\usepackage{multirow}
\usepackage{graphicx}
\usepackage[caption=false,font=footnotesize]{subfig}
\usepackage{algorithm}
\usepackage{algorithmic}
\usepackage[hidelinks]{hyperref}
\newcommand\algorithmicprocedure{\textbf{procedure}}
\newcommand{\algorithmicendprocedure}{\algorithmicend\ \algorithmicprocedure}
\makeatletter
\newcommand\PROCEDURE[3][default]{%
  \ALC@it
  \algorithmicprocedure\ \textsc{#2}(#3)%
  \ALC@com{#1}%
  \begin{ALC@prc}%
}
\newcommand\ENDPROCEDURE{%
  \end{ALC@prc}%
  \ifthenelse{\boolean{ALC@noend}}{}{%
    \ALC@it\algorithmicendprocedure
  }%
}
\newenvironment{ALC@prc}{\begin{ALC@g}}{\end{ALC@g}}
\makeatother

\usepackage{xcolor}
\usepackage{cite}

\allowdisplaybreaks

\DeclareMathOperator{\expt}{\mathbb{E}}

\usepackage{amsthm}
\usepackage{stfloats}

\newtheorem{theorem}{Theorem}
\newtheorem{lemma}{Lemma}

\newtheorem{assumption}{Assumption}

\newtheorem{cor}{Corollary}

\DeclareMathOperator*{\argmin}{\arg\!\min}
\DeclareMathOperator*{\argmax}{\arg\!\max}

\begin{document}
\title{\fontsize{22}{28}\selectfont 
Duality-Guided Graph Learning for Real-Time Joint Connectivity and Routing in LEO Mega-Constellations
}

\author{
\IEEEauthorblockN{Zhouyou Gu, $^\dagger$Jinho Choi, Tony Q. S. Quek, Jihong Park}
\thanks{
Z. Gu, T. Q. S. Quek, and J. Park  are with the Information Systems Technology and Design Pillar, Singapore University of Technology and Design, Singapore 487372 (email: \{zhouyou\_gu, tonyquek, jihong\_park\}@sutd.edu.sg).
}
\thanks{
$^\dagger$J. Choi is with the School of Electrical and Mechanical Engineering,
the University of Adelaide, Adelaide, SA 5005, Australia
(email: jinho.choi@adelaide.edu.au).
}
\thanks{Corresponding author is J. Park. Source codes will be available at {https://github.com/zhouyou-gu}.}
}

\maketitle

\begin{abstract}
Laser inter-satellite links (LISLs) of low Earth orbit (LEO) mega-constellations enable high-capacity backbone connectivity in non-terrestrial networks, but their management is challenged by limited laser communication terminals, mechanical pointing constraints, and rapidly time-varying network topologies. This paper studies the joint problem of LISL connection establishment, traffic routing, and flow-rate allocation under heterogeneous global traffic demand and gateway availability. We formulate the problem as a mixed-integer optimization over large-scale, time-varying constellation graphs and develop a Lagrangian dual decomposition that interprets per-link dual variables as congestion prices coordinating connectivity and routing decisions. 
To overcome the prohibitive latency of iterative dual updates, we propose DeepLaDu, a Lagrangian duality-guided deep learning framework that trains a graph neural network (GNN) to directly infer per-link (edge-level) congestion prices from the constellation state in a single forward pass.
We enable scalable and stable training using a subgradient-based edge-level loss in DeepLaDu. 
We analyze the convergence and computational complexity of the proposed approach and evaluate it using realistic Starlink-like constellations with optical and traffic constraints. Simulation results show that DeepLaDu achieves up to 20\% higher network throughput than non-joint or heuristic baselines, while matching the performance of iterative dual optimization with orders-of-magnitude lower computation time, suitable for real-time operation in dynamic LEO networks.
\end{abstract}
\begin{IEEEkeywords} 
Non-terrestrial networks; mega-constellations; graph learning.
\end{IEEEkeywords}

\section{Introduction}
\IEEEPARstart
Non-terrestrial networks (NTNs), comprising satellite, aerial, and spaceborne platforms, are becoming a core enabler of global Internet access in 5G-and-beyond systems by providing low-latency, wide-area connectivity that complements terrestrial infrastructure \cite{lin20215g,azari2022evolution}.
Among them, Low Earth orbit (LEO) mega-constellations gain particular prominence due to their dense spatial reuse, low propagation latency, and ability to form a high-capacity, dynamically reconfigurable backbone for non-terrestrial networks \cite{maiolinicapez2024use,he2024directtosmartphone}.
Large commercial deployments (e.g., Starlink- and OneWeb-like systems) increasingly operate as space backbone networks \cite{wu2025enhancing}.
To sustain high-capacity inter-satellite transport without spectrum congestion in NTN backbones, these constellations rely on laser inter-satellite links (LISLs) built on free-space optical (FSO) technologies \cite{elamassie2023freea}, forming a large-scale, time-varying graph whose nodes are satellites (or LCTs) and whose edges are LISLs.

\begin{figure}[!t]
  \centering
  \includegraphics[scale=0.86]{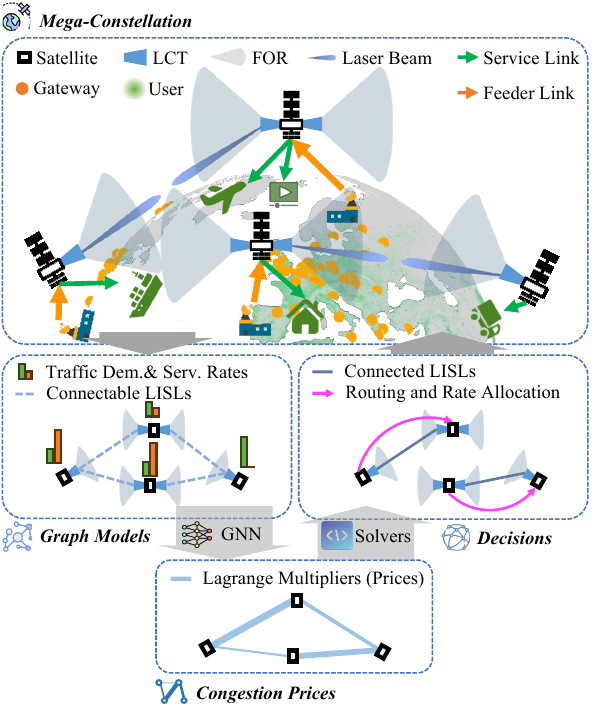}
  \vspace{-0.5cm}
  \caption{DeepLaDu for joint connectivity and routing in a mega-constellation.}
  \label{fig:leo_overview}
  \vspace{-0.3cm}
\end{figure}

A central operational challenge is to route traffic and allocate flow rates over the LISL network to meet highly uneven global demand and gateway availability \cite{bhattacherjee2019network}.
User populations and gateway locations are geographically imbalanced, so the resulting traffic demands and serving capacities vary significantly across satellites in both space and time; ignoring this heterogeneity leads to underutilized links in low-demand regions and bottlenecks near high-demand or gateway-rich areas.
Importantly, routing and rate allocation must be solved \emph{after} LISL connections are established, since feasible paths and per-link capacities depend on the realized LISL topology.

Establishing LISL connectivity is itself constrained by scarce laser communication terminals (LCTs) and their mechanics.
Unlike RF links, LISLs use highly directional beams \cite{ntontin2025vision}, so each LCT can sustain only one active link at a time, limiting the number of concurrent LISLs per satellite \cite{wang2024free}.
Moreover, beam steering is restricted by mechanical field-of-regard (FOR) and vibration/jitter effects \cite{kaymak2018survey}, further reducing feasible pairings.
These constraints are exacerbated by the time-varying LEO constellations, which dynamically change link feasibility and traffic hotspots under stringent computation time limits for reconfiguring connections and routes.

Deep learning can be a natural approach to meeting these latency constraints by learning a direct mapping from constellation state to connection, routing, and rate allocation decisions.
Recent learning-based methods have been explored for satellite-network routing and resource allocation \cite{gao2025topologycompressed,gao2025age,li2025efficient,lozano-cuadra2025continual,ran2025fullydistributed,zhang2025grlr,huang2024gnnenabled,wu2025sate}. However, large constellation size and continual topology variation make end-to-end learning of discrete connection and routing decisions difficult to scale. This motivates graph neural networks (GNNs), which naturally process variable-size, graph-structured inputs and generalize across changing connectivity patterns \cite{bengio2021machine,kool2018attentiona,nazari2018reinforcement}. Yet even with GNNs, directly learning the full joint decisions from the large decision space remains challenging as the constellation and traffic scale, due to the combinatorial connection-and-routing space coupled with continuous rate allocation.

A further opportunity is to leverage the structure of the optimization problem to decompose the decision space.
In our prior work \cite{gu2025joint}, we developed a Lagrangian dual formulation that decomposes the joint LISL connection, routing, and rate allocation problem on the constellation graph.
By relaxing per-link capacity constraints on connectable LISL edges, the dual assigns an edge-wise variable to each satellite pair, which can be interpreted as a \emph{congestion price} on that graph link.
These prices guide which satellite links should be prioritized in connection establishment and routing, such that higher prices indicate tighter edge capacity constraints and higher routing costs, thereby prioritizing link establishment while discouraging routing over the links.
Given the connection and routing decisions, the remaining flow-rate allocation reduces to a straightforward linear program \cite{gu2025joint,wu2025sate}.
Despite this favorable decomposition, updating the edge-wise congestion prices remains a bottleneck: beyond subgradient descent, there is no efficient price-update solver with comparable generality.
Subgradient methods require many iterations to approach optimal prices, which is incompatible with time-varying LEO operation under tight coherent time during which the constellation graph remains unchanged \cite{wu2025sate}.
Therefore, the key challenge is: \emph{how can one compute near-optimal edge-wise congestion prices efficiently on large, time-varying constellation graphs?}

To address this challenge, we propose a Lagrangian duality-guided deep learning framework (DeepLaDu) that uses a GNN to directly infer edge-wise congestion prices from the constellation graph, as illustrated in Fig. \ref{fig:leo_overview}.
The GNN takes the constellation graph as input and outputs congestion prices for all connectable LISL edges in a single forward pass, eliminating the need for iterative dual updates. 
Using the predicted prices, the LISL connections, traffic routing, and flow-rate allocation are computed using subproblems solvers maximum weight matching, shortest-path routing, and linear programming, respectively.
This approach is well-suited to our setting because the constellation is naturally represented as a graph with satellites as nodes and connectable LISLs as edges; the inferred congestion prices can be immediately used within our decomposition to make connection, routing, and rate-allocation decisions.

\textbf{Contributions}.\quad Our main contributions are as follows.
\begin{itemize}
  \item We formulate the joint learning task to optimize the LISL connection, traffic routing, and rate allocations by training the GNN that predicts edge-wise Lagrange multipliers (congestion prices) via one-step GNN inference. 
  The one-step inference removes the iterative dual updates required by classical Lagrangian methods \cite{fisher2004lagrangian,gu2025joint}. DeepLaDu achieves low processing time (tens of milliseconds) suitable for real-time operation in large, time-varying constellations.

  \item 
  We train the GNN using a loss function built from the dual subgradient signal, yielding direct edge-wise feedback on congestion prices.
  This differs from prior works that rely on graph-level metrics \cite{gao2025topologycompressed,gao2025age,li2025efficient,zhang2025grlr,huang2024gnnenabled} (e.g., averaged network throughput, delay or age of information) or node-level metrics \cite{lozano-cuadra2025continual,ran2025fullydistributed,wu2025sate} (per-satellite performance), which provide coarse feedback on link-level decisions and fail to converge efficiently in large constellations.

  \item We provide theoretical analysis of convergence and computational complexity for the proposed DeepLaDu framework. Under mild assumptions, we prove that DeepLaDu converges to a stationary point of the GNN parameters. The complexity analysis shows that DeepLaDu scales polynomially with the constellation size due to the low complexity solvers of the decomposed subproblems and one-step GNN inference, which is critical for real-time operation in large, time-varying mega-constellations.
  
  \item Through simulations on Starlink constellations with realistic LCT mechanics and global traffic/gateway profiles, we show that DeepLaDu improves up to 20\% network throughput over heuristic and non-joint baselines.
  We also measure the coherent time of the constellation graphs, i.e., the time duration over which the graph structure remains unchanged, showing that DeepLaDu's processing time is well within this limit. Ablation studies further demonstrate the impact of LCT and constellation configurations on the proposed method's performance, where our method constantly outperforms the baselines in the network throughput by up to $50\%$ in constellations with extreme scarce LCT resources, i.e., where each satellite has a small number of LCTs. 
  \end{itemize}

\textbf{Paper Organization}.\quad The rest of this paper is organized as follows.
Section \ref{sec:related_works} reviews the related works. Section \ref{sec:system_model} presents the system model of the LEO mega-constellation with LISLs and the graph representation of the constellation. Section \ref{sec:problem_formulation} formulates the joint optimization problem over the graphs and discusses its challenges.
Next, Sections \ref{sec:lagrangian_dual_relaxation} and \ref{sec:learning_based_dual_optimization} propose the formulation of DeepLaDu framework and its GNN/loss design, respectively.
Finally, Section \ref{sec:simulation_results} shows the simulations evaluating the proposed methods, and Section \ref{sec:conclusion} concludes this work.

\section{Related Work}\label{sec:related_works}
In the literature on LISL connection design, existing approaches typically link each satellite to immediate neighbors in the same orbit and nearby satellites in adjacent orbits \cite{chaudhry2021laser,chen2021analysis}, which forms a grid connection pattern.
Inter-orbit links can also be adjusted to alternative neighbors to enhance connectivity \cite{chen2020topology,rao2025minimumhop,guo2024constellation}. Beyond the grid pattern, optimizing motif-based patterns, small and repeatable graph structures, can improve robustness and efficiency \cite{bhattacherjee2019network}. More adaptive schemes often use maximum-weight matching (MWM), ranking potential LISLs by link metrics \cite{leyva-mayorga2021interplane,ron2025time}, with weights based on a single factor (e.g., peak rate) \cite{leyva-mayorga2021interplane} or composite scores (e.g., latency and capacity) \cite{ron2025time}.
However, the above works neglect non-uniform traffic profile from uneven user and gateway distributions. This can lead to underutilized LCTs and LISLs and suboptimal network
throughput \cite{bhattacherjee2019network}. Incorporating traffic flow routing into the LCT connection design is essential to improve the overall throughput of the constellation.
Meanwhile, existing works perform routing after deciding the topology connection using either max-flow \cite{tao2023transmitting} and Dijkstra \cite{huang2024efficient}, or distributed next-hop selection \cite{ekici2001distributed}. 
The joint design of LISL connectivity and traffic routing requires further investigation to improve LCT utilization by connecting satellites with high traffic demand.

Learning-based methods have recently been widely studied for ISL-connected constellation graphs \cite{gao2025topologycompressed,gao2025age,li2025efficient,lozano-cuadra2025continual,ran2025fullydistributed,zhang2025grlr,huang2024gnnenabled,wu2025sate}.
For instance, works in \cite{gao2025topologycompressed,gao2025age} use GNN to predict age-of-information (AoI) to send data from the source satellite node to the destination satellite node and select minimum AoI paths to route the traffic in the constellation graph \cite{gao2025topologycompressed}.
Meanwhile, distributed learning approaches train NNs operating on each satellite to dynamically select next-hop routes at each satellite node of the graph\cite{li2025efficient,lozano-cuadra2025continual,ran2025fullydistributed,zhang2025grlr}.
Given routing paths over the constellation graphs, other works \cite{huang2024gnnenabled,wu2025sate} use GNNs to allocate link rates between paths to maximize network throughput.
How to train GNNs to infer the decision from the large space on the joint connection and routing problem remains unexplored.
As mentioned before, learning the congestion prices on the links of the constellation effectively constructs a weighted graph representation of the constellation. While learning to construct graph representations of networks has been successfully applied, for example, to indicate active transmission links \cite{zhao2025generative} and interference relationships \cite{gu2024graph}, the graph learning task in mega-constellations requires further study to address the challenges of large amount of satellites and time-varying constellation geometries.

\section{System Model and Graph Representation of LEO Mega-Constellation}\label{sec:system_model}
We consider a LEO mega-constellation connected by LISLs, where each satellite is equipped with multiple LCTs. Note that while we instantiate the model using LEO satellites, the formulation applies to NTN backbone nodes with directional inter-node links and time-varying geometry.
\begin{figure}[!t]
  \centering
  \includegraphics[scale=0.75]{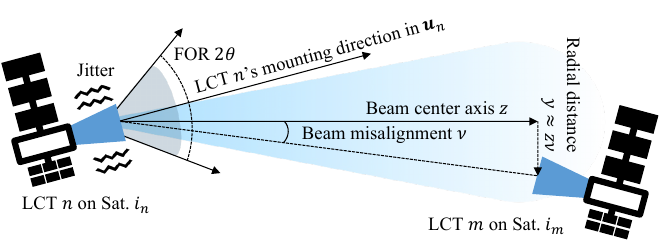}
  \vspace{-0.2cm}
  \caption{Illustration of the LCT and beam model from one LCT to another.}
  \label{fig:beam_model}
  \vspace{-0.3cm}
\end{figure}
\subsection{System Model}
\subsubsection{Orbital Dynamics}
The system time is denoted by $t$ (s), with $t=0$ aligned to a real-world time $\mathcal{T}_0$.
We adopt a Cartesian Earth-centered inertial (ECI) frame whose origin is at Earth’s center; the x-axis points toward the vernal equinox and the z-axis coincides with Earth’s spin axis \cite{vallado2022fundamentals}.
Earth rotates about the z-axis at $7.2921 \times 10^{-5}$ rad/s, and its radius is $6.3781 \times 10^6$ m.
The constellation consists of $I$ satellites indexed by $\mathcal{I}=\{1,\ldots,I\}$, each on a near-circular orbit described by two-line elements (TLEs).
Propagating the TLEs yields the ECI position $\mathbf{l}_i(t)$ (m) and velocity $\mathbf{v}_i(t)$ (m/s) of satellite $i\in\mathcal{I}$.
For any pair $(i,j)$, let $z_{i,j}(t)$ denotes the inter-satellite range and $\mathbf{d}_{i,j}(t)$ denotes the unit vector from $i$ to $j$, which are defined as
\begin{equation}
  \begin{aligned}
    z_{i,j}(t) = \|\mathbf{l}_i(t) - \mathbf{l}_j(t)\|,\ 
    \mathbf{d}_{i,j}(t) = \frac{\mathbf{l}_j(t) - \mathbf{l}_i(t)}{z_{i,j}(t)}.
  \end{aligned}
\end{equation}

\subsubsection{Satellite Form Factor, LCT Mechanics, and Beam}
Each satellite is modeled as a rigid body whose attitude is regulated by its onboard control system \cite{vallado2022fundamentals}, keeping the body frame aligned with $-\mathbf{l}_i(t)$, i.e., toward Earth’s center. Satellite $i$ carries $N'$ laser inter-satellite communication terminals (LCTs), giving $N=N'I$ terminals across the constellation indexed by $\mathcal{N}=\{1,\dots,N\}$. Terminal $n$ resides on satellite $i_n\in\{1,\dots,I\}$ and has a mounting direction $\mathbf{u}_n(t)$ (unit vector). Each LCT steers within a cone of half-angle $\theta$ (radians) about $\mathbf{u}_n(t)$ using, e.g., steering mirrors. Here, this region of the LCT can point to is referred to as its field of regard (FOR). Due to space dynamics and micro-disturbances, the realistic pointing process exhibits small random deviations (pointing jitter) from the intended direction, as illustrated in Fig. \ref{fig:beam_model}.
The pointing jitter is modeled as a Rayleigh-distributed angular error \cite{arnon2003effects} in radians with a standard deviation $\sigma_{\mathrm{J}}$ (radians).
Note that pointing errors are assumed as independent and identically distributed (i.i.d.) across LCTs.
We model each LISL transmitter as producing a Gaussian beam \cite{saleh2019fundamentalsa} with total optical power $P_0$ in watts (W). The detailed beam model in the appendix presents the LISL capacity between two LCTs $n$ and $m$ as $r_{n,m}(t)$ (Gbps) if connected at time $t$.

\subsubsection{Global Traffic Profile Model}
As this study focuses on the LCT management, we adopt a simplified traffic setting in which gateway-sourced content is delivered to users within the satellite NTN backbone's coverage over the globe.
Users are modeled as data downloaders (e.g., video streaming) over service links. 
For satellite $i$ at time $t$, let $U_i(t)$ be the random number of active users in its coverage, each requesting $D$ Gbps. Multiple ground gateways are distributed globally; if a gateway lies within a satellite’s coverage, that satellite can serve up to $Q$ Gbps of traffic. 
With gateway access, a satellite first serves its local demand $U_i(t)D$ and, via LISLs, can relay any residual capacity to other satellites. Accordingly, after serving local demand, satellite $i$ can serve other satellites' demand up to
\begin{equation}
  \begin{aligned}
    Q_i(t) = \begin{cases}
      (Q - U_i(t) D)_{+},\ &\text{if satellite $i$ has gateway access},\\
      0,\ &\text{otherwise},
    \end{cases}    
  \end{aligned}
\end{equation}
which we term the traffic serving rate of satellite $i$ to the constellation. Here $(x)_{+}=\max\{x,0\}$.
Likewise, the traffic demand at satellite $i$ is the residual user load within its coverage that is not satisfied by its own gateway (if available) and is defined as
\begin{equation}
  \begin{aligned}
    D_i(t) = 
    \begin{cases}
      (U_i(t) D-Q)_{+},&\text{if sat. $i$ has gateway access},\\
      U_i (t)D,&\text{otherwise},
    \end{cases}    
  \end{aligned}
\end{equation}
referred to as the traffic demand rate of satellite $i$.
In this setting, any positive residual demand at a satellite is served via LISLs by satellites that have gateway access.
Satellites with a positive serving/demand rate are termed serving/demanding satellites, respectively.

In the unconstrained case, any satellite with gateway access could serve a demanding satellite $s'$, creating an excessive number of source--destination pairs and inflating complexity. To limit the source--destination pairing set, we restrict $s'$ to be served via LISLs by its $M$ nearest gateway-capable satellites.
Let $\argmin_{\mathcal{S}\subseteq\mathcal{I},|\mathcal{S}|=M} \sum_{s\in\mathcal{S}} z_{s,s'}$ denote these $M$ closest candidates to $s'$, where $z_{s,s'}$ is the inter-satellite distance. 
We then define the set of all source--destination pairs as
\begin{equation}\label{eq:const:traffic:flow_pairs_init}
  \begin{aligned}
    \mathcal{F} = \{(s,s')\big|D_{s'} > 0; s \in \argmin_{\mathcal{S}\subseteq\mathcal{I},|\mathcal{S}|=M} \sum_{s\in\mathcal{S}} z_{s,s'} \}.
  \end{aligned}
\end{equation}

\begin{figure}[!t]
  \centering
  \includegraphics[scale=0.875]{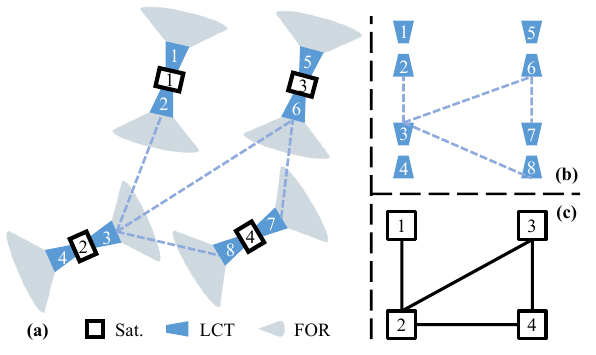}
  \caption{A constellation graph at a time instance $t$: (a) a network of four satellites, each equipped with two LCTs; (b) its LCT connectivity graph $\mathcal{G}^{\text{LCT}}(t)$; and (c) the satellite adjacency graph $\mathcal{G}^{\text{SAT}}(t)$.}
  \label{fig:graph_model}
  \vspace{-0.3cm}
\end{figure}
\subsection{Constellation Graphs}
\subsubsection{Graph Construction}
Given the constellation at a time instance $t$, we construct two separate constellation graphs to represent 1) the LCT connectivity for LISL formation and 2) satellite adjacency for traffic routing.
Specifically, we construct the LCT connectivity graph to represent all connectable LCT pairs that can potentially form a LISL.
We represent LCTs by the graph $\mathcal{G}^{\text{LCT}}(t)=(\mathcal{N},\mathcal{E}(t))$, whose nodes, $\mathcal{N}$, are LCTs and whose edges, $\mathcal{E}(t)$, connect LCT pairs that can potentially form a link at time $t$.
Two LCTs can potentially establish an LISL only when 1) the separation between their mounting satellites is below the maximum optical range $\hat{z}$ and 2) each satellite lies within the other terminal’s FOR. Accordingly, all connectable LCT pairs is expressed as
\begin{equation}
  \begin{aligned}
    &\mathcal{E}(t) = \big\{\{n,m\}\ \big| \ i_n \neq i_m, z_{i_n,i_m}(t) \leq \hat{z}, \\
    &\quad \mathbf{d}_{i_n,i_m}(t) \cdot \mathbf{u}_n(t) > \cos \theta,\mathbf{d}_{i_m,i_n}(t) \cdot \mathbf{u}_m(t) > \cos \theta \big\}.
  \end{aligned}
\end{equation}
Then, we construct the satellite adjacency graph to represent neighboring satellite pairs that can potentially connect via at least one pair of their LCTs, where the traffic can be routed. The set of all neighboring satellite pairs $(i,j)$ is denoted as
\begin{equation}
  \begin{aligned}
    \mathcal{L}(t) = \{(i,j)\ \big|\  \mathcal{E}_{i,j}(t)\neq \emptyset \},
  \end{aligned}
\end{equation}
where $\mathcal{E}_{i,j}(t)$ collects connectable LCTs between $i$ and $j$ as
\begin{equation}
  \begin{aligned}
    \mathcal{E}_{i,j}(t) = \big\{\{n,m\}\big| i_n = i, i_m = j,\{n,m\}\in\mathcal{E}(t)\big\}.
  \end{aligned}
\end{equation}
We denote $\mathcal{G}^\text{SAT}(t)=(\mathcal{I},\mathcal{L}(t))$ as the satellite adjacency graph, where nodes represent satellites and an edge $(i,j)$ exists iff at least one LCT pair between $i$ and $j$ is connectable.
Fig. \ref{fig:graph_model} illustrates the two constellation graphs, LCT connectivity graph $\mathcal{G}^{\text{LCT}}(t)$ and satellite adjacency graph $\mathcal{G}^{\text{SAT}}(t)$, when four satellites each carry two LCTs.

\begin{figure}[!t]
  \centering
  \includegraphics[scale=0.8]{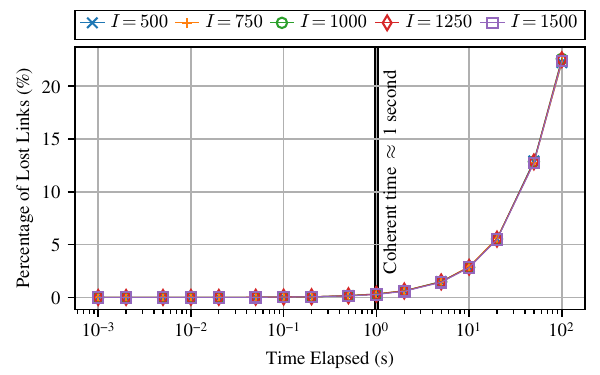}
  \caption{The percentage of lost connectable LCT pairs over time in a Starlink-like constellation with different numbers of satellites where each satellite carries $N'=2$ LCTs and has a FOR of $\theta=60$ deg (detailed constellation parameters are provided in Section \ref{sec:simulation_results}).}\label{fig:plot_test_ld_starlink_constellation_time_varying_coherent_time}
  \vspace{-0.3cm}
\end{figure}

\begin{figure*}
  \centering
  \includegraphics[scale=0.9]{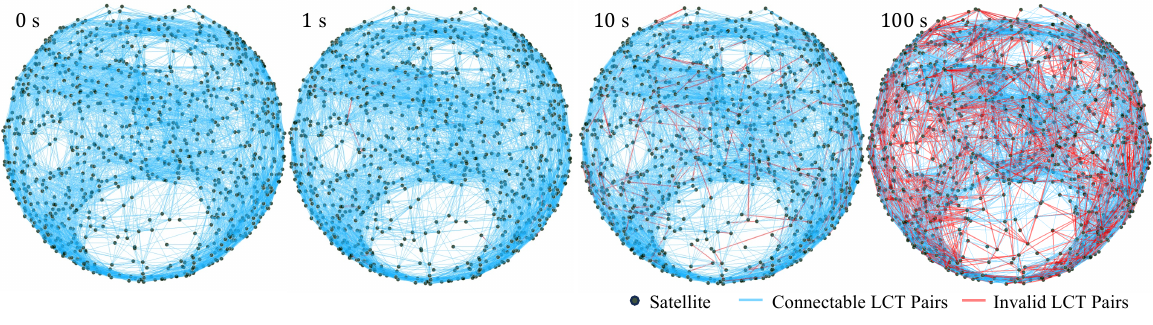}
  \vspace{-0.3cm}
  \caption{The constellation graph $\mathcal{G}^\text{LCT}$ indicating the connectable LCT pairs at the starting ($0$ s) and different ending time ($1$ s, $10$ s, $100$ s) when $I=1000$. Blue lines indicate the connectable LCT pairs, while red lines indicate the LCT pairs that are no longer connectable due to the satellite movement (detailed constellation parameters are provided in Section \ref{sec:simulation_results}).
    } \label{fig:constellation_simulation_1000_compare_time_varying}
  \vspace{-0.3cm}
\end{figure*}

\subsubsection{Time-Varying Nature of Graphs}
The constellation is constantly changing due to orbital dynamics, i.e., the relative motions between satellites.
Consequently, the connectable links between LCT pairs vary over time. For instance, Fig. \ref{fig:plot_test_ld_starlink_constellation_time_varying_coherent_time} shows the number of connectable LCT pairs over time in a Starlink-like constellation with different numbers of satellites.
We measure the coherent time of the constellation graph by counting the percentage of connectable LCT pairs that are lost compared to a random initial time instance.
As shown in Fig. \ref{fig:plot_test_ld_starlink_constellation_time_varying_coherent_time}, most of the connectable LCT pairs remain in the both LCTs' FOR within $1$ second.
This indicates that the constellation graph topology remains mostly unchanged within $1$ second.
Here, the duration of the constellation graph topology remaining mostly unchanged is referred to as the coherent time of the constellation.

Additionally, we show the constellation with $I=1000$ satellites in Fig. \ref{fig:constellation_simulation_1000_compare_time_varying} and the geographical differences in the satellites and connectable LCT pairs at the starting time and at a given time later. 
The figure shows that within $1$ second of time elapsed, the constellation structure does not change much, while after $10$ seconds and $100$ seconds, the constellation structure changes significantly due to the satellite movement. This implies that the constellation can be assumed quasi-static within a short time interval, e.g., on the order of $1$ second, and the optimization decisions should be made within this duration to adapt to the time-varying constellation.

More generally in NTNs, such time variation may arise from mobility, platform reconfiguration, or link-state dynamics, all of which manifest as changes in the underlying graph topology.
Next, we will formulate the joint optimization problem of LISL connections and traffic routing over the constellation graphs at an instantaneous time instance, and provide the solutions in the following sections. Since the problem and the solutions are discussed at each time instance, we omit the time index $t$ hereafter for brevity.

\section{Joint Connectivity and Routing Problem}\label{sec:problem_formulation}
This section formulates the joint optimization problem of LISL connections and traffic routing over the constellation graphs. Note that this formulation is directly adopted from \cite{gu2025joint}.
We first define three sets of decision variables: 1) LCT pair connection indicators over the LCT connectivity graph $\mathcal{G}^{\text{LCT}}$, 2) flow routing path selection variables over the satellite adjacency graph $\mathcal{G}^{\text{SAT}}$, and 3) flow rate allocation variables among the paths.
Then, we formulate feasibility sets for these three decision variables based on the constellation graphs and impose per-link capacity constraints that couple the decision variables.
Finally, we formulate the joint optimization problem in \eqref{eq:prob:constrainted_routing_and_matching:primal} and discuss the challenges in solving it.

\subsection{Decision Variables}
\subsubsection{LCT Pair Connection Indicators}
We denote whether LCTs $n$ and $m$ are connected via a LISL as a binary indicator $c_{n,m}\in\{0,1\}$, where $c_{n,m}=1$ indicates that LCTs $n$ and $m$ are connected.
The collection of all indicators $c_{n,m}$ is denoted as a vector $\mathbf{c} = \{c_{n,m}\}_{(n,m)\in\mathcal{E}}$.

\subsubsection{Traffic Routing Path Selection Variables}
For each flow $(s,s')\in \mathcal{F}$, we denote whether the directed link $(i,j)$ in the satellite adjacency graph $\mathcal{G}^{\text{SAT}}$ is used on the path from $s$ to $s'$ as a binary variable $x^{s,s'}_{i,j}\in\{0,1\}$, where $x^{s,s'}_{i,j}=1$ indicates that link $(i,j)$ is used.
We denote routing decisions from satellite $s$ to $s'$ as $\mathbf{x}^{s,s'} = \{x^{s,s'}_{i,j}\}_{(i,j)\in\mathcal{L}}$, and all routing decisions as $\mathbf{x} = \{\mathbf{x}^{s,s'}\}_{(s,s')\in\mathcal{F}}$.

\subsubsection{Traffic Rate Allocation Variables}
The data rate of each flow $(s,s')\in \mathcal{F}$ is denoted as $q^{s,s'}$ (Gbps), representing the rate at which satellite $s$ serves the demand of satellite $s'$ over the established LISL connections and routed paths.
We write all traffic flows as a vector $\mathbf{q} = \{q^{s,s'}\}_{(s,s')\in\mathcal{F}}$.

\subsection{Feasible Sets of Decision Variables}
We now formulate the feasible sets of the three decision variables based on the constellation graphs.
\subsubsection{LCT Connection Constraints}
Note that a LISL connection is reciprocal between the connected two LCTs, i.e.,
\begin{equation}\label{eq:const:link:binary_bidirectional_connection}
  \begin{aligned}
    c_{n,m} = c_{m,n} ,\ \forall n\neq m.
  \end{aligned}
\end{equation}
Only one LCT can be connected to each LCT at a time as
\begin{equation}\label{eq:const:link:one_connection}
  \begin{aligned}
    \sum_{m\in\mathcal{N}} c_{n,m} = 1,\ \forall n\in\mathcal{N}.
  \end{aligned}
\end{equation}
The feasible domain of LCT connections is defined as
\begin{equation}
  \begin{aligned}
    \mathcal{C} = \left\{ \mathbf{c} \ \big|\eqref{eq:const:link:binary_bidirectional_connection},\eqref{eq:const:link:one_connection}; c_{n,m}\in\{0,1\},\forall (n,m)\in\mathcal{E} \right\}.
  \end{aligned}
\end{equation}

\subsubsection{Traffic Routing Constraints}
Traffic is forwarded only across adjacent satellites.
For each source--destination pair $(s,s')$, let $x^{s,s'}_{i,j}\in\{0,1\}$ indicate whether the directed link $(i,j)$ is used on the path from $s$ to $s'$.
These routing variables obey standard flow-conservation constraints over $\mathcal{G}^\text{SAT}$ as
\begin{equation}\label{eq:const:path:flow_constraints}
\begin{aligned}
  \sum_{j:(s,j)\in \mathcal{L}} x^{s,s'}_{s,j} - \sum_{j:(j,s)\in \mathcal{L}} x^{s,s'}_{j,s} & = 1, \quad \text{(source flow)}
  \\
  \sum_{j:(s',j)\in \mathcal{L}} x^{s,s'}_{s',j} - \sum_{j:(j,s')\in \mathcal{L}} x^{s,s'}_{j,s'} & = -1, \quad \text{(destination flow)} 
  \\
  \sum_{j:(i,j)\in \mathcal{L}} x^{s,s'}_{i,j} - \sum_{j:(j,i)\in \mathcal{L}} x^{s,s'}_{j,i} & = 0, \quad \forall i \in \mathcal{I} \setminus \{s,s'\}.
\end{aligned}
\end{equation}
Let $\mathbf{x}^{s,s'} = \{x^{s,s'}_{i,j}\}_{(i,j)\in\mathcal{L}}$ denote the routing decisions from satellite $s$ to $s'$, and its feasible domain is defined as
\begin{equation}\label{eq:const:path:feasible_paths}
  \begin{aligned}
    \mathcal{X}^{s,s'} = \{\mathbf{x}^{s,s'} \ \big |  \eqref{eq:const:path:flow_constraints};\ x^{s,s'}_{i,j}\in\{0,1\},\forall (i,j)\in\mathcal{L}\}.
  \end{aligned}
\end{equation}
In addition, we write routing decisions of all flows as $\mathbf{x} = \{\mathbf{x}^{s,s'}\}_{(s,s')\in\mathcal{F}}$ that has a feasible domain denoted as
\begin{equation}
  \begin{aligned}
    \mathcal{X} = \{\mathbf{x} \ \big | \ \mathbf{x}^{s,s'} \in \mathcal{X}^{s,s'},\forall (s,s')\in\mathcal{F}\}.
  \end{aligned}
\end{equation}

\subsubsection{Traffic Serving and Demand Rates Constraints}
For each source--destination pair $(s,s')$, the flow $q^{s,s'}$ is limited by $s$’s serving capacity and $s'$’s residual demand, i.e.,
\begin{equation}\label{eq:const:traffic:flow_rate}
  \begin{aligned}
    \sum_{s':(s,s')\in\mathcal{F}}  q^{s,s'} \leq Q_s,\forall s\in\mathcal{I};\sum_{s:(s,s')\in\mathcal{F}}  q^{s,s'} \leq D_{s'}, \forall s'\in\mathcal{I}.
  \end{aligned}
\end{equation}
We define the feasible domain of traffic flows as
\begin{equation}
  \begin{aligned}
    \mathcal{Q} = \{\mathbf{q} \ \big |\  \eqref{eq:const:traffic:flow_rate};\ q^{s,s'} \geq 0,\ \forall (s,s')\in\mathcal{F}\}.
  \end{aligned}
\end{equation}

\subsection{Joint Link Matching and Flow Routing Problem}
For each satellite pair $(i,j)$, the total rate of all flows routed from $i$ to $j$ must not exceed the combined capacity of the LISLs established between them, i.e.,
\begin{equation}\label{eq:const:path:flow_rate}
  \begin{aligned}
    \sum_{(s,s')\in\mathcal{F}} q^{s,s'} x^{s,s'}_{i,j} \leq \sum_{\{n,m\}\in\mathcal{E}_{i,j}} r_{n,m} c_{n,m},\ \forall (i,j)\in\mathcal{L},
  \end{aligned} 
\end{equation}
which couples the three decision variables $\mathbf{c}$, $\mathbf{x}$, and $\mathbf{q}$.
Given the above feasible domain definitions and the link-capacity constraint \eqref{eq:const:path:flow_rate}, we formulate the joint link matching and flow routing task as maximizing the total network throughput, $\sum_{(s,s')\in\mathcal{F}} q^{s,s'}$, as
\begin{equation}\label{eq:prob:constrainted_routing_and_matching:primal}
  \begin{aligned}
  \min_{
    \substack{
      \mathbf{c}\in\mathcal{C};\mathbf{x}\in\mathcal{X};\mathbf{q}\in\mathcal{Q}
    }
  }
  -\sum_{(s,s')\in\mathcal{F}} q^{s,s'}, \ \text{s.t.}\ \eqref{eq:const:path:flow_rate}.
  \end{aligned}
\tag{\bf{P1}}
\end{equation}
Here, $\mathcal{C}$, $\mathcal{X}$, and $\mathcal{Q}$ denote the feasibility sets for LCT matching, routing decisions and flow allocations, respectively. For convenience, we express the objective in minimization form by negating the network throughput, i.e., maximizing $\sum_{(s,s')} q^{s,s'}$ is equivalent to minimizing $-\sum_{(s,s')} q^{s,s'}$. The resulting formulation is a mixed-integer program (MIP) and is NP-hard.
Thus, solving \eqref{eq:prob:constrainted_routing_and_matching:primal} optimally is computationally prohibitive for large constellations. Moreover, the problem must be solved repeatedly as the constellation evolves over time.
To address these challenges, we propose a learning-based Lagrangian dual optimization framework, namely DeepLaDu, to efficiently solve \eqref{eq:prob:constrainted_routing_and_matching:primal} in real-time, as detailed next.


\section{Learning Lagrange Multipliers on Constellation Graph as Congestion Prices}\label{sec:lagrangian_dual_relaxation}
This section presents the Lagrangian dual relaxation of the joint optimization problem in \eqref{eq:prob:constrainted_routing_and_matching:primal} that interprets the Lagrange multipliers as congestion prices over the constellation graphs to coordinate the three decision blocks, i.e., LCT matching, routing, and flow rate allocation.
The relaxation and interpretation on Lagrange multipliers are adopted from \cite{gu2024graph}.
Built upon this Lagrangian dual relaxation, we formulate a learning task in which a GNN is trained to predict the multipliers from the states of the constellation graphs.

\subsection{Lagrangian Dual Relaxation}\label{subsec:lagrangian_dual_relaxation}
Examining \eqref{eq:prob:constrainted_routing_and_matching:primal}, only the per-link capacity constraints \eqref{eq:const:path:flow_rate} couple the three decision blocks, LCT matching, flow-rate allocation, and routing, whereas the remaining requirements live independently in the feasibility sets $\mathcal{C}$, $\mathcal{X}$ and $\mathcal{Q}$.
Introducing Lagrange multipliers for \eqref{eq:const:path:flow_rate} moves these coupling constraints into the objective, which separates the three blocks while letting them interact through the multipliers.
Specifically, the multipliers penalize those links with violations in their maximum link rate constraints, which indicates that the link should be prioritized to be connected, but should be avoided in traffic flows' routing paths.
By updating the multipliers, one can coordinate matching, rate allocation, and routing jointly, yet solve them as independent subproblems.

\begin{figure*}[t]
\begin{align}
      g(\lambda)&= \min_{
        \substack{
          \mathbf{c}\in\mathcal{C};\mathbf{x}\in\mathcal{X};\mathbf{q}\in\mathcal{Q}
        }
      }
      L(\mathbf{c}, \mathbf{x}, \mathbf{q}, \lambda) = \min_{
        \substack{
          \mathbf{c}\in\mathcal{C};\mathbf{x}\in\mathcal{X};\mathbf{q}\in\mathcal{Q}
        }
      } -\sum_{(s,s')\in\mathcal{F} } q^{s,s'} + \sum_{(i,j)}\lambda_{i,j} (\sum_{(s,s')\in \mathcal{F}} q^{s,s'} x^{s,s'}_{i,j} - \sum_{\{n,m\}\in\mathcal{E}_{i,j}} r_{n,m} c_{n,m}) \notag
      \\
      &= \min_{
        \substack{
           \mathbf{x}\in \mathcal{X};\mathbf{q}\in\mathcal{Q}
        }
      }
      \Big\{
      \sum_{(s,s')\in\mathcal{F} } q^{s,s'} (-1 + \sum_{(i,j)\in\mathcal{L}}\lambda_{i,j} x^{s,s'}_{i,j})   
      \Big\} 
      +\min_{
        \substack{
          \mathbf{c}\in\mathcal{C}
        }
      }
      \Big\{
      -\sum_{\{n,m\}\in\mathcal{E}} (\lambda_{i_n,i_m} + \lambda_{i_m,i_n})  r_{n,m} c_{n,m}
      \Big\} \notag
      \\
      &= 
      -\underbrace{\max_{
        \substack{
          \mathbf{q}\in \mathcal{Q}
        }
      }
      \Big\{
      \sum_{(s,s')\in\mathcal{F} } q^{s,s'} (1 - 
      \overbrace{\min_{
        \substack{
          \mathbf{x}^{s,s'}\in\mathcal{X}^{s,s'}
        }
      }
      \sum_{(i,j)\in\mathcal{L}}\lambda_{i,j} x^{s,s'}_{i,j}}^{
      \substack{\text{(b) Minimum cost routing}}
      })
      \Big\}}_{\substack{\text{(c) Linear weighted flow rate maximization}}}
      -
      \overbrace{\max_{
        \substack{
          \mathbf{c}\in\mathcal{C}
        }
      }
      \Big\{
      \sum_{\{n,m\}\in\mathcal{E}} (\lambda_{i_n,i_m} + \lambda_{i_m,i_n})  r_{n,m} c_{n,m}
      \Big\}
      }^{
      \substack{\text{(a) Maximum weight LISL matching}}
      }. \label{eq:prob:constrainted_routing_and_matching:lagrangian_dual}
  \end{align}
  \hrule
\vspace{-0.5cm}
\end{figure*}

Specifically, by relaxing the maximum link rate constraints \eqref{eq:const:path:flow_rate}, the Lagrangian-augmented objective of \eqref{eq:prob:constrainted_routing_and_matching:primal} is
\begin{equation}\label{eq:prob:constrainted_routing_and_matching:lagrangian}
  \begin{aligned}
      &L(\mathbf{c}, \mathbf{x}, \mathbf{q}, \lambda) = -\sum_{(s,s')\in\mathcal{F} } q^{s,s'} \\
      &\qquad + \sum_{(i,j)}\lambda_{i,j} \Big(\sum_{(s,s')\in \mathcal{F}} q^{s,s'} x^{s,s'}_{i,j} - \sum_{\{n,m\}\in\mathcal{E}_{i,j}} r_{n,m} c_{n,m} \Big) ,
  \end{aligned}
\end{equation}
where $\lambda=\{\lambda_{i,j}\}_{(i,j)\in\mathcal{L}}$ is the set of Lagrange multipliers (or dual variables) for each link rate constraint \eqref{eq:const:path:flow_rate} in all neighboring satellite pairs $(i,j)\in\mathcal{L}$, with $\lambda_{i,j}\geq 0$.
The above Lagrangian-augmented objective relaxes \eqref{eq:prob:constrainted_routing_and_matching:primal} as
\begin{equation}\label{eq:prob:constrainted_routing_and_matching:relaxed}
  \begin{aligned}
  \{ \hat{\mathbf{c}}(\lambda) , \hat{\mathbf{x}}(\lambda), \hat{\mathbf{q}}(\lambda) \}= \argmin_{
      \substack{
        \mathbf{c}\in\mathcal{C};\mathbf{x}\in\mathcal{X};\mathbf{q}\in\mathcal{Q}
      }
    }
    L(\mathbf{c}, \mathbf{x}, \mathbf{q},  \lambda),
  \end{aligned}
  \tag{\textbf{P2}}
\end{equation}
where $\hat{\mathbf{c}}(\lambda)$, $\hat{\mathbf{q}}(\lambda)$ and $\hat{\mathbf{x}}(\lambda)$ denotes the optimal decisions given $\lambda$ minimizing the Lagrangian-augmented objective.
Given $\lambda$, the optimal value of the Lagrangian-augmented objective $L(\mathbf{c}, \mathbf{x}, \mathbf{q}, \lambda)$ in \eqref{eq:prob:constrainted_routing_and_matching:relaxed} is a function of $\lambda$, referred to as the dual function $g(\lambda)$.
The dual problem is to maximize the dual function $g(\lambda)$ over the Lagrange multipliers $\lambda$ as
\begin{equation}\label{eq:prob:constrainted_routing_and_matching:dual}
  \begin{aligned}
    \max_{\lambda\geq 0} g(\lambda) = \max_{\lambda\geq 0} \min_{
      \substack{
        \mathbf{c}\in\mathcal{C};\mathbf{x}\in\mathcal{X};\mathbf{q}\in\mathcal{Q}
      }
    }
    L(\mathbf{c}, \mathbf{x}, \mathbf{q}, \lambda).
  \end{aligned}
\tag{\textbf{P3}}
\end{equation}
As the dual function $g(\lambda)$ is always concave, the dual problem \eqref{eq:prob:constrainted_routing_and_matching:dual} is a convex optimization problem, and its maximum value provides a lower bound to the optimal value of the original problem \eqref{eq:prob:constrainted_routing_and_matching:primal} \cite{fisher2004lagrangian,boyd2004convex}.

\subsection{Converting Lagrange Multipliers to Matching/Routing}\label{subsec:dual_variable_guided_matching_and_routing}
The Lagrange multipliers from the dual problem \eqref{eq:prob:constrainted_routing_and_matching:dual} penalize violations in the link rate constraints \eqref{eq:const:path:flow_rate}, i.e., when traffic is overloaded on a link beyond its established LISL capacity, the multiplier for that link increases to discourage further routing over it, and the connection matching prioritizes connecting that link to establish more capacity.
This behavior allows us to interpret the optimized multipliers as congestion prices over the constellation graphs, guiding the matching and routing decisions.
Given Lagrange multipliers $\lambda$, each $\lambda_{i,j}$ value as a congestion price on the link-rate constraint in \eqref{eq:const:path:flow_rate} for the neighboring pair $(i,j)\in\mathcal{L}$. A larger $\lambda_{i,j}$ indicates a tighter (harder-to-satisfy) link rate constraint and heavier load on $(i,j)$. 
Accordingly, LISL matching should prioritize connecting the LCTs of $i$ and $j$ to expose capacity, whereas routing should de-emphasize traversing $(i,j)$ to avoid further congestion. Leveraging this interpretation, a feasible primal solution can be constructed from the optimized multipliers by: 1) forming max weight matching (MWM) of LCT pairs using $\lambda$ as weights, 2) computing minimum-cost routes for all source–destination satellites over the connected links weighted with $\lambda$, and 3) allocating flow rates based on the selected routes and the connected LISLs to maximize network throughput.
From an NTN perspective, the congestion prices act as global coordination signals for scarce backbone resources under the link capacity constraints.

Specifically,  the greedy weight matching is first used to compute the MWM of the graph $\mathcal{G}^{\text{LCT}}(\mathcal{N},\mathcal{E})$ with the given Lagrange multipliers $\lambda$ as
\begin{equation}\label{eq:rounding:mwm}
  \begin{aligned}
    \tilde{\mathbf{c}} \leftarrow \mathrm{MWM}(\mathcal{N}, \mathcal{E}, \{\lambda_{i_n,i_m} \cdot r_{n,m}\}_{(n,m)\in \mathcal{E}}).
  \end{aligned}
\end{equation}
Collect the connected satellite pairs in the above as 
\begin{equation}
  \begin{aligned}
    \mathcal{L}' =\{(i,j)| \exists (n,m) \in \mathcal{E}_{i,j}, c_{n,m} = 1 \}.
  \end{aligned}
\end{equation}
Then, we compute the routing paths based on the connected constellation graph $\mathcal{G}'^{\text{SAT}}=(\mathcal{I}, \mathcal{L}')$ using Dijkstra's algorithm, where Lagrange multipliers $\lambda_{i,j}$ weight the edges as
\begin{equation}\label{eq:rounding:spf}
  \begin{aligned}
    \tilde{\mathbf{x}}^{s,s'} \leftarrow \mathrm{SPF}(s,s',\mathcal{I}, \mathcal{L}', \{\lambda_{i,j}\}_{(i,j)\in \mathcal{L}'}), \ \forall (s,s')\in\mathcal{F}.
  \end{aligned}
\end{equation}
Note that in the connected constellation, not all source--destination satellite pairs $(s,s')\in\mathcal{F}$ have a routing path, i.e., $\tilde{\mathbf{x}}^{s,s'}$ may be infeasible in \eqref{eq:rounding:spf} for some $(s,s')\in\mathcal{F}$.
We remove those infeasible source--destination satellite pairs and define the set of connected source--destination pairs as
\begin{equation}
  \begin{aligned}
    \mathcal{F}' = \{(s,s')\in\mathcal{F} \,|\, \tilde{\mathbf{x}}^{s,s'}\ \text{is feasible in}\ \text{\eqref{eq:rounding:spf}}\},
  \end{aligned}
\end{equation}
reducing the number of routing paths to compute the traffic flow rates.
Finally, the flow rates will be allocated based on the connected constellation and routing paths by solving the flow rate maximization (FRM) problem as
\begin{equation}\label{eq:rounding:frm}
  \begin{aligned}
    \tilde{\mathbf{q}} \leftarrow \argmax_{
        \substack{
          \mathbf{q}\in \mathcal{Q}
        }
      }
      \sum_{(s,s')\in\mathcal{F}'}  q^{s,s'}, \ \text{s.t.}\ \eqref{eq:const:path:flow_rate}\ \text{given}\ \hat{\mathbf{c}},\ \tilde{\mathbf{x}}^{s,s'}.
  \end{aligned}
\end{equation}
Here, the constraints in \eqref{eq:const:path:flow_rate} are linear constraints on $\mathbf{q}$ given the previously computed $\hat{\mathbf{c}}$ and $\tilde{\mathbf{x}}^{s,s'}$, and thus the problem \eqref{eq:rounding:frm} is a linear program that can be solved efficiently.

With given Lagrange multipliers $\lambda$, the above three steps \eqref{eq:rounding:mwm}, \eqref{eq:rounding:spf}, and \eqref{eq:rounding:frm} convert the multipliers to a feasible primal solution $\{\tilde{\mathbf{c}}, \tilde{\mathbf{x}}, \tilde{\mathbf{q}}\}$ of the original problem \eqref{eq:prob:constrainted_routing_and_matching:primal}. In other words, we reduce the decision space of \eqref{eq:prob:constrainted_routing_and_matching:primal} from $\{\mathbf{c}, \mathbf{x}, \mathbf{q}\}$ to $\lambda$, where the space size is reduced from $\mathcal{O}(| \mathcal{E} | \cdot |\mathcal{F}| \cdot | \mathcal{L} |)$ to $\mathcal{O}(| \mathcal{L} |)$.
To solve the dual problem \eqref{eq:prob:constrainted_routing_and_matching:dual} and find the optimal Lagrange multipliers $\lambda^*$, we can use the subgradient descent method \cite{boyd2003subgradient} to update the Lagrange multipliers iteratively. However, as the constellation is highly dynamic, the dual problem needs to be solved in real-time, which is challenging for the traditional subgradient descent method due to its iterative nature.
Therefore, we propose to train a GNN to directly approximate the optimal Lagrange multipliers $\lambda^*$ via one forward propagation, which can be used to make real-time decisions on Lagrange multipliers and further on the LCT connections, routing paths and flow rates.

\subsection{Learning Lagrange Multipliers on Constellation Graphs}
The learning task is to train a GNN to approximate the optimal Lagrange multipliers $\lambda^*$ that maximizes the dual function $g(\lambda)$ over all possible constellation structures and traffic serving and demand rates.
\subsubsection{Constellation State as GNN Input}
The GNN operates on the satellite neighboring graph $\mathcal{G}^\text{SAT}$. The graph's adjacency matrix collects link capacities of all connectable LCT pairs between satellites and is defined as
\begin{equation}
  \begin{aligned}
    \mathbf{R} = 
    \big[R_{i,j} | R_{i,j} = 
    \begin{cases}
          \sum_{\{n,m\}\in\mathcal{E}_{i,j}} r_{n,m} ,\ \forall (i,j)\in\mathcal{L},\\
          0, \text{ otherwise.}
    \end{cases}
    \big],
  \end{aligned}
\end{equation}
which is the edge feature matrix of the graph.
The node features of the graph are the traffic serving and demand rates of the satellites, which is defined as
\begin{equation}
  \begin{aligned}
    \mathbf{s}_i = [Q_i, D_i]^{\rm T},\ \forall i\in\mathcal{I};\ \mathbf{S} = [\mathbf{s}_1,\dots,\mathbf{s}_N]^{\rm T},
  \end{aligned}
\end{equation}
where $\mathbf{S}$ collects all node features of the graph.

\subsubsection{GNN Output as Lagrange Multipliers}
The output of the GNN is the Lagrange multipliers $\lambda_{i,j}$ for all neighboring satellite pairs $(i,j)\in\mathcal{L}$, which is defined as
\begin{equation}
  \begin{aligned}
    \mathbf{\lambda} \approx \mu(\mathbf{S}, \mathbf{R} | \mathbf{w}),
  \end{aligned}
\end{equation}
where $\mu(\cdot | \mathbf{w})$ is the GNN with NN weight parameters $\mathbf{w}$.
Note that the Lagrange multipliers are non-negative, which can be ensured by applying a non-negative activation function at the output layer of the GNN. 
Furthermore, we can prove that the optimal Lagrange multipliers can be upper bounded by a finite value, which can be formally stated as follows.

\begin{lemma}
\label{lemma:bounded_lagrange_multipliers}
There exist the optimal Lagrange multipliers $\lambda^* = \{\lambda^*_{i,j}\}_{(i,j)\in\mathcal{L}}$ of \eqref{eq:prob:constrainted_routing_and_matching:dual}, $\lambda^* \in \argmax_{\lambda\geq 0} g(\lambda)$, such that $0 \leq \lambda^*_{i,j} \leq 1$, $\forall (i,j)\in\mathcal{L}$.
\begin{proof}
The proof is listed in appendix.
\end{proof}
\end{lemma}

The above lemma indicates that we can apply a bounded non-negative activation function, e.g., the sigmoid function, at the output layer of the GNN to ensure the output Lagrange multipliers are bounded between 0 and 1.

\subsubsection{GNN Learning Task Formulation}
The GNN learning task is to train the GNN to approximate the optimal Lagrange multipliers that maximizes the dual function $g(\lambda|\mathbf{S},\mathbf{R})$ for all possible constellation structure and traffic serving and demand rates, i.e., $\mathbf{S}$ and $\mathbf{R}$. 
Mathematically, the task can be formulated as an optimization problem as
\begin{equation}\label{eq:prob:gnn_training:original}
  \begin{aligned}
    \max_{\mathbf{w}} \expt_{(\mathcal{G}^\text{SAT},\mathcal{G}^\text{LCT})\sim\Gamma}\left[g\left(\mu(\mathbf{S}, \mathbf{R} | \mathbf{w})|\mathcal{G}^\text{SAT},\mathcal{G}^\text{LCT}\right) \right],
  \end{aligned}
\end{equation}
where $\Gamma$ is the distribution over all constellation setups, and $g(\cdot|\mathcal{G}^\text{SAT},\mathcal{G}^\text{LCT})$ is the dual function given the Lagrange multipliers returned from the GNN, $\mu(\mathbf{S}, \mathbf{R} | \mathbf{w})$.
The detail design of the GNN architecture, the loss and the training algorithm to solve the above learning task are presented in the next.

\section{GNN and Subgradient-Based Loss Function Design in DeepLaDu}\label{sec:learning_based_dual_optimization}
This section details the GNN architecture, the loss function and the training algorithm to solve the learning task in \eqref{eq:prob:gnn_training:original}.

\subsection{GNN Architecture}
\begin{figure}[t]
  \centering
  \includegraphics[scale=1.]{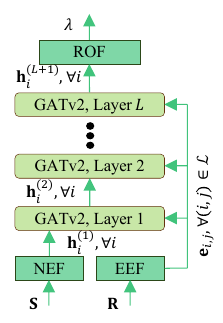}
  \vspace{-0.3cm}
  \caption{Illustration of the GNN architecture, $\mu(\mathbf{S}, \mathbf{R} | \mathbf{w})$, and its forward propagation and backpropagation processes. Here, NEF, EEF, and ROF represents the MLPs, namely node embedding function, edge embedding function, and read-out function, respectively.
  }
  \label{fig:gnn_architecture}
  \vspace{-0.3cm}
\end{figure}

We use a graph attention network (GAT), e.g., GATv2 \cite{brody2021how}, to build the GNN architecture, as illustrated in Fig. \ref{fig:gnn_architecture}.
It has learnable attention weights adaptive to the importance of neighboring nodes and edges during message passing, unlike uniform neighbor aggregation in traditional GNNs.

In detail, we first encode the node features $\mathbf{s}_i$ and the edge features $R_{i,j}$ into initial node embeddings $\mathbf{h}_i^{(0)}$ and edge embeddings $\mathbf{e}_{i,j}$ via learnable MLPs as
\begin{equation}
  \begin{aligned}
    \mathbf{h}_i^{(1)} = \mathrm{MLP}_\mathrm{NEF}(\mathbf{s}_i | \mathbf{w}_{\mathrm{NEF}} );\ \mathbf{e}_{i,j} = \mathrm{MLP}_\text{EEF}(R_{i,j}| \mathbf{w}_{\mathrm{EEF}}).
  \end{aligned}
\end{equation}
The GNN has $L$ layers, where the $l$-th layer, $l=1,\dots,L$, updates each node's embedding $\mathbf{h}_i^{(l-1)}$ to $\mathbf{h}_i^{(l)}$ by aggregating messages from its neighbors using multi-head GATv2 \cite{brody2021how} as
\begin{equation}\label{eq:gnn_gat_weighted_aggregation}
  \begin{aligned}
    \mathbf{h}_i^{(l+1)} = \sigma_\mathrm{ReLU}\Big(\frac{1}{U}\sum_{u=1}^U\sum_{j\in\mathcal{N}(i)} \alpha_{i,j}^{(l,u)} \mathbf{w}_{\mathrm{msg}}^{(l,u)} \mathbf{h}_j^{(l)}\Big),
  \end{aligned}
\end{equation}
where $\mathcal{N}(i) = \{j | (i,j)\in\mathcal{L}\}$ is the set of neighbors of node $i$ in the graph, $\sigma_\mathrm{ReLU}$ is the ReLU activation function, and $U$ is the number of attention heads.
Here, $\alpha_{i,j}^{(l,u)}$ in \eqref{eq:gnn_gat_weighted_aggregation} is the attention weight on the message from node $j$ to $i$ at layer $l$ and head $u$, computed as
\begin{equation}\label{eq:gnn_gat_attention_weight}
  \begin{aligned}
    \tilde{\alpha}_{i,j}^{(l,u)}= \mathbf{w}_{\mathrm{cmb}}^{(l,u)} \sigma_\mathrm{LkRe}(\mathbf{w}_{\mathrm{src}}^{(l,u)} \mathbf{h}_i^{(l,u)} \!+\! \mathbf{w}_{\mathrm{dst}}^{(l,u)} \mathbf{h}_j^{(l,u)} \!+\! \mathbf{w}_{\mathrm{edg}}^{(l,u)} \mathbf{e}_{j,i}),
  \end{aligned}
\end{equation}
where $\sigma_\mathrm{LkRe}$ is the Leaky ReLU activation function. The attention weights $\tilde{\alpha}_{i,j}^{(l,u)}$ are then normalized via softmax as
\begin{equation}
  \begin{aligned}
    \alpha_{i,j}^{(l,u)} = \frac{\exp(\tilde{\alpha}_{i,j}^{(l,u)})}{\sum_{k\in\mathcal{N}(i)} \exp(\tilde{\alpha}_{i,k}^{(l,u)})}. 
  \end{aligned}
\end{equation}
Here, $\mathbf{w}_{\mathrm{cmb}}^{(l,u)}$, $\mathbf{w}_{\mathrm{src}}^{(l,u)}$, $\mathbf{w}_{\mathrm{dst}}^{(l,u)}$, $\mathbf{w}_{\mathrm{edg}}^{(l,u)}$, and $\mathbf{w}_{\mathrm{msg}}^{(l,u)}$ in \eqref{eq:gnn_gat_attention_weight} are learnable parameters in GATv2 at layer $l$.
The final output of the GNN is computed via a readout function (ROF) MLP on the last layer's node embeddings as
\begin{equation}
  \begin{aligned}
    \lambda_{i,j} = \mathrm{MLP}_{\mathrm{ROF}}(\mathbf{h}_i^{(L+1)},\  \mathbf{h}_j^{(L+1)} |\ \mathbf{w}_{\mathrm{ROF}}),\ \forall (i,j)\in\mathcal{L}.
  \end{aligned}
\end{equation}
The union of parameters in node and edge feature embedding functions, GATv2 layers, and the readout functions is the GNN parameters, i.e., 
$\mathbf{w}\allowbreak=\allowbreak\cup_{l,u} \{\mathbf{w}_{\mathrm{msg}}^{(l,u)}, \mathbf{w}_{\mathrm{cmb}}^{(l,u)}, \mathbf{w}_{\mathrm{src}}^{(l,u)}, \mathbf{w}_{\mathrm{dst}}^{(l,u)}, \mathbf{w}_{\mathrm{edg}}^{(l,u)}\}\cup\mathbf{w}_{\mathrm{NEF}}\cup\mathbf{w}_{\mathrm{EEF}}\cup\mathbf{w}_{\mathrm{ROF}}$.
The detailed architecture of the MLP is listed in Section \ref{sec:simulation_results}. 
Note that the readout function's last layer uses a sigmoid activation function to ensure the output Lagrange multipliers are in the range $[0,1]$, as stated in Lemma \ref{lemma:bounded_lagrange_multipliers}.

\subsection{Subgradient-Based Learning Algorithm}

\begin{algorithm}[!t]
\caption{DeepLaDu Training Procedure for Joint Link Matching and Traffic Routing}\label{alg:gnn_training}
\begin{algorithmic}[1]
\STATE Initialize GNN parameters $\mathbf{w}$ randomly.
\STATE Initialize constellation/traffic distribution $\Gamma$.
\FOR{ $k = 1,2,\dots,K$}
  \STATE Sample $\mathcal{G}^\text{SAT}, \mathcal{G}^\text{LCT}$ from the distribution $\Gamma$.
  \STATE Extract node, edge features $\mathbf{S}$, $\mathbf{R}$ from $\mathcal{G}^\text{SAT}$, $\mathcal{G}^\text{LCT}$.
  \STATE Compute the multipliers as $\lambda = \mu(\mathbf{S}, \mathbf{R} |\mathbf{w})$.
  \STATE Compute the subgradient $\delta(\lambda)$ in \eqref{eq:dual_function_subgradient} as \eqref{eq:routine:mwm_given_lambda}\eqref{eq:routine:spf_given_lambda}\eqref{eq:routine:flow_rate_given_routing_cost}.
  \STATE Update the GNN parameters as \eqref{eq:gnn_parameter_update}.
  \STATE Update the learning rate as \eqref{eq:learning_rate_update}.
\ENDFOR
\STATE \textbf{Return} the trained GNN parameters $\mathbf{w}$.
\end{algorithmic}
\end{algorithm}

\subsubsection{Subgradient-Based Loss Function}
Define the loss function accordingly to maximize the objective in \eqref{eq:prob:gnn_training:original} as
\begin{equation}\label{eq:loss:gnn_training}
  \begin{aligned}
    L(\mathbf{w}) = -\expt_{(\mathcal{G}^\text{SAT},\mathcal{G}^\text{LCT})\sim\Gamma}\left[g\left(\mu(\mathbf{S}, \mathbf{R} | \mathbf{w})|\mathcal{G}^\text{SAT},\mathcal{G}^\text{LCT}\right)\right].
  \end{aligned}
\end{equation}
The dual function $g(\lambda|\mathcal{G}^\text{SAT},\mathcal{G}^\text{LCT})$ is non-differentiable with respect to $\lambda$ as it involves solving a mixed-integer programming problem \eqref{eq:prob:constrainted_routing_and_matching:relaxed}.
However, we can use the subgradient \cite{boyd2004convex,boyd2003subgradient} of the dual function to approximate the gradient of the loss function $L(\mathbf{w})$ with Lagrange multipliers, e.g.,
\begin{equation}\label{eq:loss:gnn_training:gradient_approximation}
  \begin{aligned}
    &\nabla_{\mathbf{w}} L(\mathbf{w})\\
    =& \expt_{(\mathbf{S},\mathbf{R})\sim\Gamma}\Big[-
    \sum_{i,j}\frac{\partial g(\lambda|\mathbf{S},\mathbf{R})}{\partial \lambda_{i,j}} \nabla_{\mathbf{w}}\lambda_{i,j}|_{\lambda=\mu(\mathbf{S}, \mathbf{R} | \mathbf{w})} \Big] \\
    \approx& \expt_{(\mathbf{S},\mathbf{R})\sim\Gamma}[-\delta(\lambda)_{i,j}\nabla_{\mathbf{w}}\lambda_{i,j} |\lambda=\mu(\mathbf{S}, \mathbf{R} | \mathbf{w})],\\
  \end{aligned}
\end{equation}
where the derivative of the dual function $\frac{\partial g(\lambda|\mathbf{S},\mathbf{R})}{\partial \lambda_{i,j}}$ is approximated by its subgradient w.r.t. $\lambda_{i,j}$, $\delta(\lambda)_{i,j}$, $\forall (i,j)\in\mathcal{L}$.

\subsubsection{Computing Subgradient via Decomposed Subproblems}
As the dual function $g(\lambda|\mathcal{G}^\text{SAT},\mathcal{G}^\text{LCT})$ is a piecewise linear function with respect to $\lambda$ with subgradient computed as \cite{boyd2004convex}
\begin{equation}\label{eq:dual_function_subgradient}
  \begin{aligned}
   &\delta(\lambda)_{i,j}=\!\!\!
   \sum_{(s,s')\in\mathcal{F}}\!\!\! \hat{q}^{s,s'}(\lambda) \hat{x}^{s,s'}_{i,j}(\lambda) - \!\!\!\!\!\!\sum_{\{n,m\}\in\mathcal{E}_{i,j}} \!\!\!\!r_{n,m} \hat{c}_{n,m}(\lambda),
  \end{aligned}
\end{equation}
where $\hat{q}^{s,s'}(\lambda)$, $\hat{x}^{s,s'}_{i,j}(\lambda)$, and $\hat{c}_{n,m}(\lambda)$ are the optimal decisions given $\lambda$ in \eqref{eq:prob:constrainted_routing_and_matching:relaxed}.
To obtain these optimal decisions, we observe that the relaxed problem \eqref{eq:prob:constrainted_routing_and_matching:relaxed} can be decomposed into three independent subproblems, which can be solved separately, as shown in \eqref{eq:prob:constrainted_routing_and_matching:lagrangian_dual}.
Exploiting this structure, we solve the three subproblems in \eqref{eq:prob:constrainted_routing_and_matching:lagrangian_dual} sequentially as follows.

First, the part (a) in \eqref{eq:prob:constrainted_routing_and_matching:lagrangian_dual} is an MWM problem on $\mathcal{G}^\text{LCT}=(\mathcal{N},\mathcal{E})$ with edge weights $ \{\lambda_{i_n,i_m}r_{n,m}\}_{\{n,m\}\in\mathcal{E}}$, solved as
\begin{equation}\label{eq:routine:mwm_given_lambda}
  \begin{aligned}
    \hat{\mathbf{c}}(\lambda) = \mathrm{MWM}(\mathcal{N}, \mathcal{E}, \{\lambda_{i_n,i_m}r_{n,m}\}_{\{n,m\}\in\mathcal{E}}).
  \end{aligned}
\end{equation}
Here, we use the greedy weight matching algorithm \cite{deligkas2017computational} to approximately solve the MWM problem at a low computational complexity, which sequentially matches the LCT pairs from higher weights to lower ones.
Next, the shortest path routing algorithm, Dijkstra's algorithm, is used to solve the minimum cost routing problem in the part (b).
Note that the routing decision $\mathbf{x}^{s,s'}$ for each source--destination satellite pair $(s,s')\in\mathcal{F}$ is independent of each other, and thus can be solved in parallel on the graph $\mathcal{G}^\text{SAT}=(\mathcal{I}, \mathcal{L})$ as
\begin{equation}\label{eq:routine:spf_given_lambda}
  \begin{aligned}
    \hat{\mathbf{x}}^{s,s'}(\lambda) = \mathrm{SPF}(s,s',\mathcal{I}, \mathcal{L}, \lambda),\ \forall (s,s')\in\mathcal{F}.
  \end{aligned}
\end{equation}
With the above routing decisions, we can compute the routing costs for each source--destination satellite pair $(s,s')\in\mathcal{F}$ as $\sum_{(i,j)\in\mathcal{L}}\lambda_{i,j} \hat{x}^{s,s'}_{i,j} (\lambda)$.
Based on the routing costs, we can compute the FRM problem by solving the part (c) in \eqref{eq:prob:constrainted_routing_and_matching:lagrangian_dual} as
\begin{equation}\label{eq:routine:flow_rate_given_routing_cost}
  \begin{aligned}
    \hat{\mathbf{q}}(\lambda) = \argmax_{
        \substack{
          \mathbf{q}\in\mathcal{Q}
        }
      }
      \sum_{(s,s')\in\mathcal{F} } q^{s,s'} (1 - \sum_{(i,j)\in\mathcal{L}}\lambda_{i,j} \hat{x}^{s,s'}_{i,j} (\lambda)),
  \end{aligned}
\end{equation}
where the objective is to maximize the total flow rates of all source--destination satellite pairs $(s,s')\in\mathcal{F}$, and the routing costs penalize the flow rates on high-cost traffic flows.
Here, all constraints in $\mathcal{Q}$ are linear constraints, and the objective is linear in $\mathbf{q}$, which implies that the problem \eqref{eq:routine:flow_rate_given_routing_cost} is a linear program and that can be solved efficiently.

\begin{figure}[!t]
  \centering
  \includegraphics[scale=0.795]{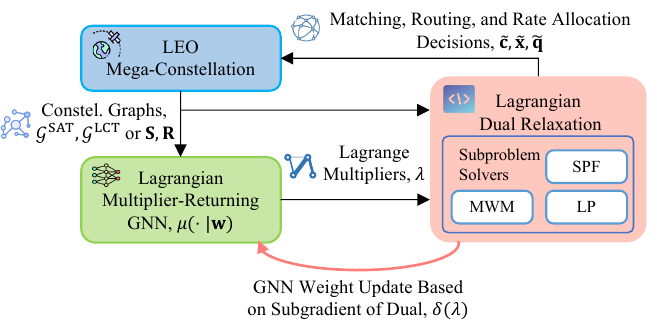}
  \vspace{-0.3cm}
  \caption{Illustration of the interaction on the GNN in DeepLaDu with the Lagrangian dual relaxation and the constellation.}
  \label{fig:ladugl_framework}
  \vspace{-0.3cm}
\end{figure}

\subsubsection{Summary of GNN Training Algorithm}
The training algorithm of the GNN is summarized in Algorithm \ref{alg:gnn_training}, and the interaction among the GNN, the Lagrangian dual relaxation, and the constellation is illustrated in Fig. \ref{fig:ladugl_framework}.
In each iteration, we first sample the constellation structure and traffic serving and demand rates at random. Then, we compute the Lagrange multipliers using the GNN, and further compute the subgradient of the dual function $g(\lambda)$ as \eqref{eq:dual_function_subgradient}. The approximated gradient in \eqref{eq:loss:gnn_training:gradient_approximation} is then used to update the GNN parameters as
\begin{equation}\label{eq:gnn_parameter_update}
  \begin{aligned}
    \mathbf{w}^{[k+1]} = \mathbf{w}^{[k]} + \alpha^{[k]} \sum_{(i,j)}
    \delta(\lambda)_{i,j} \nabla_{\mathbf{w}^{[k]}}\lambda_{i,j}|_{\lambda=\mu(\mathbf{S}, \mathbf{R} | \mathbf{w}^{[k]})} ,
  \end{aligned}
\end{equation}
where $\alpha^{[k]}$ is the learning rate at the $k$-th iteration and $\mathbf{w}^{[k]}$ are the GNN parameters at the $k$-th iteration.
We assume that the learning rate is decaying as
\begin{equation}\label{eq:learning_rate_update}
  \begin{aligned}
    \alpha^{[k]} = \frac{\alpha_0}{k^\beta}, \ k = 1,2,\dots,K.
  \end{aligned}
\end{equation}
$0<\alpha_0<1$ is a constant and $\beta $ is the decaying rate in $[0.5, 1)$.

\section{Convergence and Complexity of DeepLaDu}\label{sec:convergence_and_complexity_analysis}
\subsection{Convergence Analysis}
The convergence of the GNN training algorithm relies on the following assumptions.
\begin{assumption}
$\sum_{(i,j)}\delta(\lambda)_{i,j} \nabla_{\mathbf{w}}\lambda_{i,j}|_{\lambda=\mu(\mathbf{S}, \mathbf{R} | \mathbf{w})}$ is an unbiased estimator of $\nabla_{\mathbf{w}} L(\mathbf{w})$ with a bounded variance as
\begin{equation}
  \begin{aligned}
    &\expt[\sum_{(i,j)}\delta(\lambda)_{i,j} \nabla_{\mathbf{w}}\lambda_{i,j}|_{\lambda=\mu(\mathbf{S}, \mathbf{R} | \mathbf{w})}] \approx \nabla_{\mathbf{w}} L(\mathbf{w});\\
    &\expt[\|\sum_{(i,j)}\delta(\lambda)_{i,j} \nabla_{\mathbf{w}}\lambda_{i,j}|_{\lambda=\mu(\mathbf{S}, \mathbf{R} | \mathbf{w})} - \nabla_{\mathbf{w}} L(\mathbf{w})\|^2] \leq \sigma_L^2,
\end{aligned}
\end{equation}
\end{assumption}
\begin{assumption}
The loss function $L(\mathbf{w})$ is $\gamma$-Lipschitz smooth with respect to the GNN parameters $\mathbf{w}$.
\end{assumption}
Then, the convergence of the GNN training algorithm is guaranteed with the decaying learning rate as
\begin{theorem}\label{theorem:convergence:gnn_training}
The GNN gradient in Algorithm \ref{alg:gnn_training} converges to a stationary point, i.e., $ \lim_{k\to\infty} \min \{\|\nabla_{\mathbf{w}^{[i]}} L(\mathbf{w}^{[i]})\|^2\}_{i=1}^{k} = 0$, when $0.5\leq\beta<1 $ with a convergence rate as $\mathcal{O}(k^{-(1-\beta)})$.
\end{theorem}
\begin{proof}
  The proof relies on the assumptions above and the boundedness of optimal Lagrange multipliers in Lemma \ref{lemma:bounded_lagrange_multipliers}, and is listed in appendix.
\end{proof}
Theorem \ref{theorem:convergence:gnn_training} establishes that the proposed DeepLaDu training procedure converges to a stationary point of the duality-guided learning objective, despite the non-differentiability induced by the underlying mixed-integer structure. This result is significant for two reasons. First, it provides a formal justification for replacing iterative dual updates with a single-pass GNN inference, showing that the learned congestion prices remain theoretically well-behaved and stable under standard stochastic approximation conditions. Second, the convergence guarantee directly links learning dynamics to the physical and network constraints of the constellation through bounded Lagrange multipliers, ensuring that the learned prices retain their economic interpretation as link-level congestion signals. As a result, DeepLaDu embeds principled optimization structure into the learning process of mega-constellation networks, where convergence and decision interpretability are critical.

\subsection{Complexity Analysis}
For simplicity, we denote the number of connectable LCT pairs as $E=|\mathcal{E}|$, and the number of neighboring satellite pairs is upper bounded as $|\mathcal{L}|\leq E$ since each LCT pair connects two neighboring satellites.
\begin{cor}\label{cor:complexity:dual_optimization}
    1) The complexity of DeepLaDu in Algorithm \ref{alg:gnn_training} is $\mathcal{O}(K(E\log E + I E\log N  + \mathrm{poly}(I)))$ for $K$ iterations. 2) The complexity of converting the optimized Lagrange multipliers to the matching, routing, and rate allocation decisions in Section \ref{subsec:dual_variable_guided_matching_and_routing} is $\mathcal{O}(E\log E + I E\log N  + \mathrm{poly}(I+E))$.
\end{cor}
\begin{proof}
We first analyze the complexity of MLPs and GATv2 in the GNN. 
Configuring hidden layer dimensions are proportional to either input or output dimensions, MLPs have a complexity is quadratic to their input and output dimensions \cite{goodfellow2016deep}.
In our case, MLPs have constant dimensions regardless of the number of satellites $I$ or LCTs $N$. 
Thus, we write their complexity as $\mathcal{O}(1)$ for each forward propagation or backpropagation.
The complexity of the GATv2 layer in \eqref{eq:gnn_gat_weighted_aggregation} and \eqref{eq:gnn_gat_attention_weight} is $\mathcal{O}(E)$ for forward propagation and backpropagation, and $E$ is the upper bound on the number of edges \cite{brody2021how}.

The number of source--destination satellite pairs is bounded as $|\mathcal{F}| \leq I \cdot M\approx\mathcal{O}(I)$.
For approximating the MWM using the greedy weight matching in \eqref{eq:routine:mwm_given_lambda}, we need to sort the edge weights, which takes $\mathcal{O}(E\log E)$ complexity. 
When matching the LCT pairs, we need to iterate through all edges in the constellation and check whether the LCT pairs have been matched before, which approximately takes $\mathcal{O}(E)$ complexity. 
Dijkstra's algorithm for computing the shortest path in \eqref{eq:routine:spf_given_lambda} takes $\mathcal{O}(|\mathcal{F}|\cdot E\log N )\approx\mathcal{O}(I E\log N )$ for all $|\mathcal{F}|$ source--destination satellite pairs \cite{cormen2022introduction}.
The linear programming problem in \eqref{eq:routine:flow_rate_given_routing_cost} can be efficiently solved by the simplex method taking polynomial time at the number of constraints in $\mathcal{Q}$, i.e., $\mathcal{O}(\mathrm{poly}(I))$ \cite{vershynin2009hirsch}.
Collecting these complexities, we have the complexity on each iteration of the GNN training algorithm in Algorithm \ref{alg:gnn_training} as $\mathcal{O}(E + E\log E + I E\log N  + \mathrm{poly}(I))$, and the overall complexity for $K$ iterations is stated.
The complexity of converting the optimized Lagrange multipliers can be derived analogously.
\end{proof}

This complexity analysis confirms that DeepLaDu achieves a favorable scalability profile for large, time-varying LEO constellations. The overall cost grows polynomially with the constellation size and is dominated by standard graph operations, avoiding the combinatorial explosion inherent in solving the original mixed-integer problem directly. More importantly, the dependence on the iteration count $K$ is confined to the offline GNN training phase, while online decision-making requires only a single forward inference followed by one round of matching, routing, and rate allocation. As a result, DeepLaDu effectively transforms an otherwise iterative and latency-intensive dual optimization process into a one-shot, computationally efficient procedure, making it well suited for real-time LISL reconfiguration within the short coherent time of dynamic mega-constellation graphs.

\section{Simulation Results}\label{sec:simulation_results}
This section evaluates the proposed methods.
\subsection{Baseline Methods}
\begin{itemize}
  \item \textbf{End-to-end Learning (PG, DDPG)}: Reinforcement learning (RL) is compared to train the GNN to make link matching and traffic routing decisions. Note that due to the large decision space of link matching and traffic routing in the primal problem \eqref{eq:prob:constrainted_routing_and_matching:primal}, it is difficult to directly train the GNN to output the matching and routing decisions. Thus, we use the GNN to output the Lagrange multipliers for all neighboring satellite pairs, similar to our proposed method, and then convert the multipliers to the matching and routing decisions. The RL algorithms use the network throughput as the reward. Specifically, two cases of RL are compared: \textbf{PG:} the policy gradient (PG) \cite{sutton1999policy} assumes that the GNN returns a logit-normal distribution of the Lagrange multipliers over $[0,1]$. The multipliers are sampled from the distribution and the backpropagation is performed on the log-derivative trick to estimate the gradient of the network throughput w.r.t. each satellite pair's Lagrange multiplier. \textbf{DDPG:} the deep deterministic PG (DDPG) \cite{lillicrap2019continuous} with an additional GNN (a critic) is used to approximate the network throughput for the given constellation state and multipliers, referred to as an actor-critic algorithm. Here, the gradient w.r.t. multipliers is estimated using back-propagation on the critic and then the multiplier-returning GNN (the actor).
  \item \textbf{Joint Optimization Using Iterative Subgradient Decsent (LaDu)}: We use the iterative subgradient descent method designed in our previous work \cite{gu2025joint} to optimize the dual variables individually for a given instance of the constellation state, namely LaDu. The LaDu method iteratively updates the dual variables based on the subgradient of the dual function until a given number, $K$, of iterations is reached. We write the LaDu with $K$ iterations as LaDu-$K$. Note that LaDu optimizes the dual variables for each constellation state from with an initialization $\lambda_{i,j}=1$ $\forall (i,j)$, unlike the proposed DeepLaDu method that learns a generalizable GNN to predict the dual variables using one forward inference.
  \item \textbf{Heuristic Methods (MRate, +Grid, Rand)}: We use heuristic methods to perform link matching and traffic routing without optimizing the dual variables. Specifically, we first match the LCT pairs based on a heuristic link matching method and then route the traffic flows using the weighted SPF algorithm. Finally, we compute the maximum flow rates based on the matched LCT pairs and routed paths. The heuristic link matching and routing methods include: \textbf{+Grid:} A grid-based link matching \cite{bhattacherjee2019network} prioritizes the LCT pairs that are more aligned in their pointing directions, e.g., pairs with higher alignment $\mathbf{d}_{i_n,i_m} \cdot \mathbf{u}_{n,m}+\mathbf{d}_{i_m,i_n} \cdot \mathbf{u}_{n,m}$. \textbf{Rand}: Random link matching randomly selects the LCT pairs to match, which is achieved by randomly setting the weights of LCT pairs and then applying the matching. \textbf{MRate:} Maximum link rate matching sets the weights of LCT pairs as their maximum transmission rates, $r_{n,m}$, as the objective is to maximize the network throughput. For compared heuristics, we use open shortest path first (OSPF) routing \cite{moy1998ospf} where the flows are routed using the link weights that are the reciprocal of aggregate rates of satellite pairs $1/\sum_{\{n,m\}\in \mathcal{E}_{i,j}} c_{n,m}r_{n,m}$.
  \item \textbf{Non-Joint Optimization (SaTE)}: 
  The non-joint optimization in SaTE \cite{wu2025sate} is compared, which solves the traffic rate allocations after links are matched and traffic is routed using heuristic methods, e.g., using +Grid and SPF.
  Note that SaTE only optimizes the traffic flow rates to maximize the network throughput, while the link matching and traffic routing are not jointly optimized.
  The maximum network throughput given the topology and routes is computed optimally using the linear programming solver \cite{wu2025sate}.
\end{itemize}

\subsection{Simulation Setup}
The simulations run on a workstation with an Intel Core Ultra 9 285K (24 cores) and 32 GB of memory and an NVIDIA RTX 5090 GPU.
We emulate the constellation using Starlink TLEs from CelesTrak~\cite{CelesTrak}, snapshot at $\mathcal{T}_0=$ UTC 2025-07-16 16{:}00. To vary the scale of the constellation and emulate different states, we uniformly sample $I$ satellites from the dataset. Each satellite mounts $N'=2$ LCTs oriented along and against the satellite's velocity vector, with attitude control preserving these mounting directions. Optical and receiver parameters follow~\cite{kaymak2018survey}: aperture $A=0.01~\mathrm{m}^2$, responsivity $\Psi=0.5~\mathrm{A/W}$, and RMS noise current $\sigma_{\mathrm{N}}=3{\times}10^{-7}~\mathrm{A}$~\cite{maxim_spf_transimpedance}. We set transmit power $P_0=20$~W, bandwidth $B=1$~GHz, wavelength $1.55~\mu\mathrm{m}$, and beam divergence $100~\mu\mathrm{rad}$, which yield waist $W_0=9.87{\times}10^{-3}$~m and Rayleigh range $z_{\mathrm{R}}=1.97{\times}10^{3}$~m~\cite{saleh2019fundamentalsa}. Pointing jitter is fixed at $\sigma_{\mathrm{J}}=10~\mu\mathrm{rad}$ and the outage probability threshold at $\epsilon=10^{-3}$. Each LCT has a FOR $\theta=60^\circ$, and links are connectable only between satellites with a distance up to $\hat{z}=3000$~km.
Traffic demand is derived from real-world population data~\cite{schiavina2023ghspop} within a $\approx 200$~km coverage per satellite.
We assume $0.01\%$ of real-world population are active; thus $U_i$ is modeled as Poisson with mean equal to the covered population. Each user requests $D=0.1$~Gbps. We place $100$ gateways by sampling SatNOGS sites~\cite{satnogs}. When a gateway is visible, a satellite can source up to $Q=20$~Gbps. Any residual demand is routed toward the $M=5$ nearest gateway-connected satellites, which defines the serving–demand pairs $\mathcal{F}$ in \eqref{eq:const:traffic:flow_pairs_init}. Unless explicit mentioned, we set the number of satellites as $I=1000$.

The NEF and EEF are both configured as a single linear layer with $64$ hidden units. The ROF is configured as a $3$-layer MLP with $64$ hidden units in each layer. 
The GATv2 layers are configured with $4$ attention heads and $64$ hidden units in each head.
The activation function for all MLPs is ReLU, except for the last layer of the ROF using a sigmoid function to ensure that it outputs Lagrange multipliers within $[0,1]$.
The learning rate configurations are $\alpha_0=10^{-3}$ and $\beta=0.7$ unless otherwise specified.

\subsection{Validation of Lemma \ref{lemma:bounded_lagrange_multipliers}}
\begin{figure}[!t]
  \centering
  \includegraphics[scale=0.8]{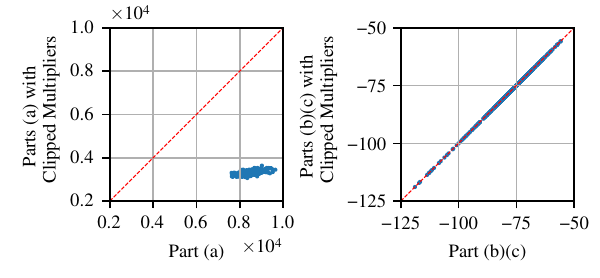}
  \vspace{-0.2cm}
  \caption{Dual function part (a), (b) and (c) values  when clipping the random Lagrange multipliers within $[0, 1]$, validating Lemma \ref{lemma:bounded_lagrange_multipliers}.}\label{fig:plot_test_lemma_1}
  \vspace{-0.3cm}
\end{figure}

We first validate Lemma \ref{lemma:bounded_lagrange_multipliers} by randomly generating the Lagrange multipliers $\lambda$ for all links $(i,j)\in\mathcal{L}$ using a log-normal distribution with its underlying normal distribution equipped with a mean $0$ and a variance $1$. Then, we clip the Lagrange multipliers within $[0, 1]$ as $\lambda'_{i,j} = \min\{1, \lambda_{i,j}\}$, $\forall (i,j)\in\mathcal{L}$, and compute the dual function values $g(\lambda)$ and $g(\lambda')$ by solving \eqref{eq:prob:constrainted_routing_and_matching:relaxed} with $\lambda$ and $\lambda'$ respectively. The value of the part (a) and part (c) in \eqref{eq:prob:constrainted_routing_and_matching:lagrangian_dual} are plotted for both $\lambda$ and $\lambda'$ in Fig. \ref{fig:plot_test_lemma_1}, where we randomly generate $100$ samples of the constellation. 
The results show that the weighted sum of flow rates in part (c) remains the same for $\lambda$ and $\lambda'$, while the MWM value in part (a) with $\lambda'$ is always less than or equal to that with $\lambda$. This validates Lemma \ref{lemma:bounded_lagrange_multipliers} that clipping the Lagrange multipliers within $[0,1]$ will not decrease the dual function value. It further implies that the optimal Lagrange multipliers are upper bounded by $1$.

\subsection{Performance of the Proposed DeepLaDu Method}
\begin{figure}[!t]
  \centering
  \includegraphics[scale=0.8]{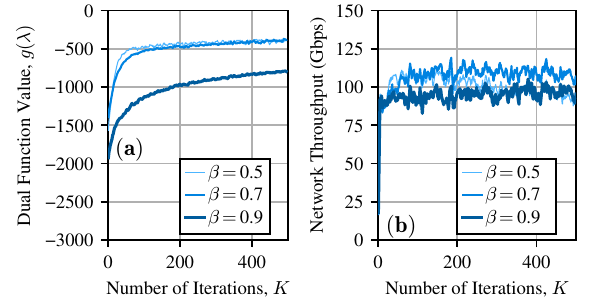}
  \vspace{-0.2cm}
  \caption{The convergence of DeepLaDu with different configuration of $\beta$}\label{fig:plot_ld_starlink_1000_varying_beta}
  \vspace{-0.25cm}
\end{figure}

\begin{figure}[!t]
  \centering
  \includegraphics[scale=0.8]{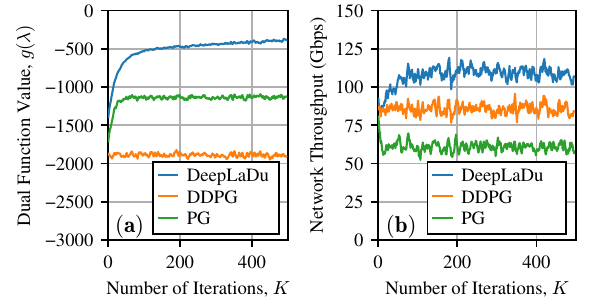}
  \vspace{-0.2cm}
  \caption{Comparison in different learning methods, DeepLaDu, PG, and DDPG.}\label{fig:plot_starlink_1000_compare_learning}
  \vspace{-0.3cm}
\end{figure}

We then evaluate the performance of our proposed DeepLaDu method in Algorithm \ref{alg:gnn_training}. 
Fig. \ref{fig:plot_ld_starlink_1000_varying_beta} shows the convergence of our method with different decaying rates $\beta$ of the learning rate. The results show that our method converges to a stationary point within $400$ iterations, and a smaller $\beta$ leads to a faster convergence rate in the dual function $g(\lambda)$, which is consistent with Theorem \ref{theorem:convergence:gnn_training}. However, when $\beta$ is too small, the method updates the GNN parameters too aggressively, leading to decreasing performance in terms of the network throughput.
In the remaining simulations, we set $\beta=0.7$. We compare the performance of our method with other learning methods, including the PG and DDPG methods as introduced above. The results in Fig. \ref{fig:plot_starlink_1000_compare_learning} show that our method outperforms the PG and DDPG methods in maximizing the dual function value and the network throughput.
This is because the PG and DDPG methods only receive the reward signal as the network throughput, which is an aggregated value of all satellite pairs and thus provides less information for guiding the learning of dual variables for each satellite pair. 
Meanwhile, our method optimizes the dual function directly based on satellite-pairwise subgradients, which provides direct feedback on the multipliers of each satellite pair and thus leads to better performance.

\begin{figure}[t]
  \centering
  \includegraphics[scale=0.725]{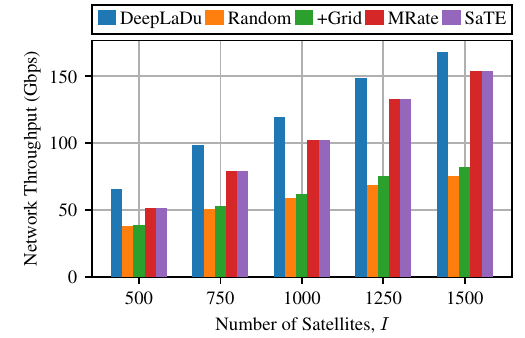}
  \vspace{-0.2cm}
  \caption{Comparison with heuristic methods.}
  \label{fig:plot_test_dual_optimization_different_constellation}
  \vspace{-0.3cm}
\end{figure}

\begin{figure}[t]
  \centering
  \includegraphics[scale=0.725]{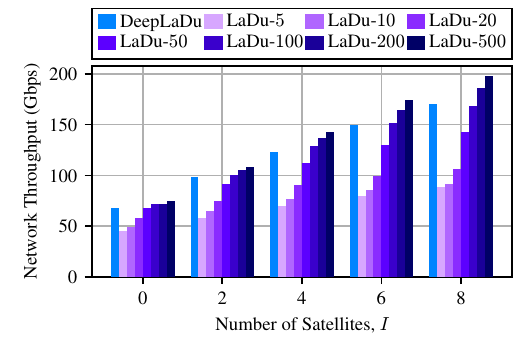}
  \vspace{-0.2cm}
  \caption{Comparison with subgradient descent method, LaDu.}
  \label{fig:plot_test_ld_starlink_1000_constellation_varying_n_sat_vs_sg}
  \vspace{-0.25cm}
\end{figure}

Using the trained GNN in our DeepLaDu method, we further evaluate the performance of our method in constellations with different sizes, i.e., different number of satellites $I$, when compared with subgradient descent method, LaDu, and heuristic methods. In Fig. \ref{fig:plot_test_dual_optimization_different_constellation}, we compare the performance of our method with heuristic/non-joint methods. The results show that our method outperforms all heuristic/non-joint methods in constellation sizes approximately $20\%\sim 100\%$ improvement in the network throughput. 
This is because our method jointly optimizes the link matching, traffic routing, and flow rates based on the constellation state, while other methods separate the link matching and traffic routing without considering their interactions and ignore the uneven distribution of traffic demand and serving rates over the globe. We also compare the performance of our method with the LaDu method in Fig. \ref{fig:plot_test_ld_starlink_1000_constellation_varying_n_sat_vs_sg}. The results show that our method achieves a comparable performance as the LaDu method with $K=100$ iterations, while our method requires only one forward pass of the GNN to predict the dual variables. This demonstrates the effectiveness of our learning-based dual optimization method in predicting the optimal dual variables without iterative optimization.

\subsection{Performance under Time-Varying Constellation}

\begin{figure}[!t]
  \centering
  \includegraphics[scale=0.8]{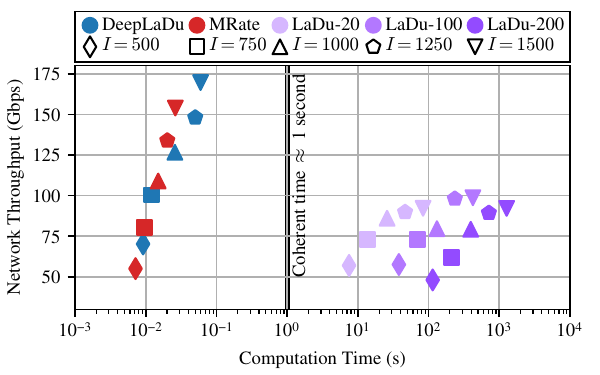}
  \vspace{-0.2cm}
  \caption{Comparison when constellation is time-varying.}
  \label{fig:plot_test_ld_starlink_1000_constellation_thr_tim_single}
  \vspace{-0.3cm}
\end{figure}
To further evaluate the performance of our method in handling time-varying constellations, we simulate the constellation dynamics and assume that the constellation structure and traffic serving and demand rates change over time while computing the dual variables, link matching, traffic routing, and flow rates.
The performance of each scheme are measured by the network throughput achieved at the solution is returned from the scheme and applied to the constellation that has evolved for the computing time of the scheme.
Fig. \ref{fig:plot_test_ld_starlink_1000_constellation_thr_tim_single} shows the network throughput and computing time of the proposed DeepLaDu method, the LaDu method with different number of iterations, and the MRate heuristic method.
The results show that the DeepLaDu outperforms the MRate with an approximately $20\%$ improvement in the network throughput while achieving a comparable computing time. This is because our GNN is trained to predict the optimal multipliers and requires only one forward pass to obtain the multipliers. On the other hand, the MRate method, even though with a short computing time, ignores traffic profiles.
When compared with the LaDu method, our DeepLaDu method improves $75\%$ in the network throughput when LaDu runs for $100$ iterations, while achieving a significantly less computing time approximately $10^{-4}$ times of LaDu-$100$.
This is because the LaDu method requires multiple iterations to optimize the dual variables from scratch for each constellation state, including solving the subproblems and computing the subgradients in each iteration, while our DeepLaDu method directly predicts the dual variables using the trained GNN.
The above results demonstrate the effectiveness of our method in handling time-varying constellations.

\subsection{Ablation Study on Constellation Configurations}
\begin{figure}[t]
  \centering
  \includegraphics[scale=0.725]{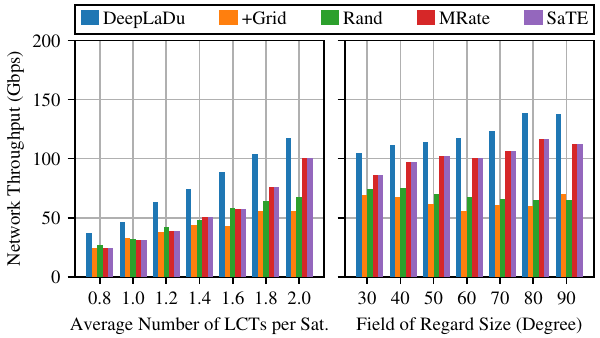}
  \vspace{-0.2cm}
  \caption{Comparison when varying LCT configurations on satellites.}
  \label{fig:plot_test_dual_optimization_vs_grid_rand_varying_availability_and_for}
  \vspace{-0.3cm}
\end{figure}

\begin{figure}[t]
  \centering
  \includegraphics[scale=0.725]{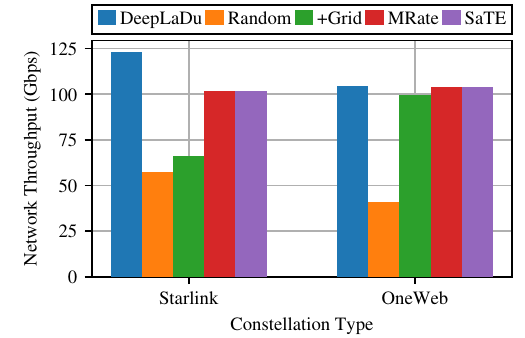}
  \vspace{-0.2cm}
  \caption{Comparison for different constellation types.}
  \label{fig:plot_test_ld_starlink_1000_constellation_varying_constellation}
  \vspace{-0.3cm}
\end{figure}

We vary LCT configurations on each satellite, of which network throughput is shown in Fig. \ref{fig:plot_test_dual_optimization_vs_grid_rand_varying_availability_and_for}.
The results show that the network throughput increases with the average number of LCTs on each satellite. This is expected since more LCTs allow more LISLs to be established.
Our method consistently outperforms the baseline methods. The maximum link rate matching method achieves the relatively higher network throughput with more LCTs. This is because high rate LISLs are more likely to be established when there are more LCTs.
Similarly, the network throughput increases when the FOR increases due to more options in the FOR. The performance of grid-based scheme is not varying much with the changing FOR, since only LCTs in the center of the steering range will be connected.
Also, we vary the pointing jitter and beam angular spreading, which influence the LISL capacity. The results in Fig. \ref{fig:plot_test_ld_starlink_1000_varying_jitter} show that the network throughput decreases when the jitter decreases. This is because a larger jitter leads to a lower LISL capacity due to the misalignment loss, as modeled in Section \ref{sec:system_model}. Moreover, we show that for a given jitter, there is an optimal beam angular spreading that maximizes the network throughput. This is because a small beam angular spreading leads to a pointing misalignment loss due to the jitter, while a large beam angular spreading leads to a larger beam power divergence loss. Thus, there exists a trade-off in designing the beam angular spreading to maximize the LISL capacity and the network throughput. From the simulation results, we observe that the optimal beam angular spreading is around $10$ times the pointing jitter.
Moreover, we compare the performance of our method in different constellation types, including Starlink ($I=1000$) and OneWeb ($I=650$) constellations. The results in Fig. \ref{fig:plot_test_ld_starlink_1000_constellation_varying_constellation} show that our method perform better than the baseline methods in the Starlink constellation while achieving a comparable performance in the OneWeb constellation. This is because the Starlink constellation has a denser satellite distribution and thus more options in link matching and traffic routing, which allows our method to better optimize the network throughput. On the other hand, the OneWeb constellation has a more sparse satellite distribution due to its higher orbital altitude and thus fewer options in link matching and traffic routing, which limits the performance gain of our method.

We further quantitively estimate the satellite coherent time of different constellations in Table \ref{tab:constellation_coherent_time}, which is estimated based on the average time duration that the connectable LCT pairs mostly remain connectable over that duration, e.g., here, we set the threshold ratio (TR) as $99.9\%$ and $99\%$ of the connectable LCT pairs at the starting time remain connectable over that duration.
The results show that the Starlink constellation has a shorter coherent time (roughly on the order of $1$ second) due to its lower orbital altitude and thus faster satellite movement, while the OneWeb constellation has a longer coherent time due to its higher orbital altitude and thus slower satellite movement. We also observe that the Kuiper constellation has a coherent time higher than Starlink, which is consistent with its orbital altitude.

\begin{table}[t]
\centering
\caption{Estimated Coherent Time of Different Constellations.}\label{tab:constellation_coherent_time}
\begin{tabular}{c|c|c}
\hline
Constellation & Coher. Time, TR=$99.9\%$ & Coher. Time, TR=$99\%$ \\
\hline
Starlink & 0.52 & 3.60 \\
OneWeb   & 0.60 & 7.18 \\
Kuiper   & 2.70 & 6.69\\
\hline
\end{tabular}
\end{table}

\begin{figure}[t]
  \centering
  \includegraphics[scale=0.725]{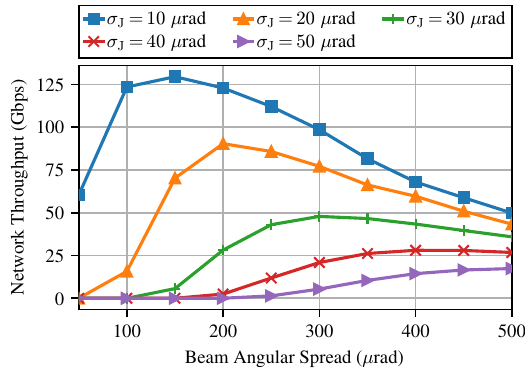}
  \vspace{-0.2cm}
  \caption{Comparison when varying the jitter and the beam angular spreading.}
  \label{fig:plot_test_ld_starlink_1000_varying_jitter}
  \vspace{-0.3cm}
\end{figure}

\section{Conclusion}\label{sec:conclusion}
This paper addressed the problem of real-time laser inter-satellite link management in large-scale, time-varying LEO mega-constellations by proposing DeepLaDu, a Lagrangian dual–guided graph learning framework that learns edge-wise congestion prices on constellation graphs. By leveraging Lagrangian dual decomposition, the proposed approach unifies LISL connection establishment, traffic routing, and rate allocation through interpretable dual variables, and replaces iterative subgradient optimization with one-step GNN inference. 
Beyond LEO mega-constellations, the proposed approach provides a scalable and interpretable approach for real-time resource coordination in non-terrestrial network backbones.
The resulting method achieves substantial throughput gains over heuristic and non-joint baselines while meeting the stringent computational constraints imposed by constellation dynamics, and is shown to scale polynomially with constellation size and to operate well within the graph coherent time. Extensive simulations on realistic Starlink-like constellations validate both the effectiveness and robustness of the approach under heterogeneous traffic demand, gateway distribution, and realistic LCT mechanics. Future work may extend this framework to incorporate uncertainty in traffic prediction and link availability, online or continual learning across evolving constellation states, and tighter integration with higher-layer network control, as well as explore alternative GNN architectures and distributed implementations suitable for onboard execution.

\section*{Appendix: LISL Capacity with Pointing Jitter}
For simplicity, we also omit the time index in this appendix.
The pointing jitter $\nu$ is modeled as a Rayleigh distribution \cite{kaymak2018survey}, with the probability density function (PDF) given by
\begin{equation}\label{eq:pointing_jitter}
  \begin{aligned}
    f_{\text{jitter}}(\nu) = \frac{\nu}{\sigma^2_{\mathrm{J}}} \exp\left(-\frac{\nu^2}{2\sigma^2_{\mathrm{J}}}\right), \nu \geq 0.
  \end{aligned}
  \notag
\end{equation}
Given the Gaussian beam model \cite{saleh2019fundamentalsa}, the beam intensity at a cross-section located a distance $z$ from the waist and at radial distance $y$ from the beam center is
\begin{equation}
  \Phi(y,z)=\Phi_0\left(\frac{W_0}{W(z)}\right)^2 \exp\!\Big(-\frac{2y^2}{W(z)^2}\Big),
  \notag
\end{equation}
where $\Phi_0$ denotes the on-axis intensity at the waist, $W_0$ is the waist radius, and $W(z)$ is the beam radius at distance $z$. 
The on-axis intensity and the beam radius evolve as
\begin{equation}
  \Phi_0 = \frac{2P_0}{\pi W_0^2}, \qquad
  W(z)=W_0\sqrt{1+\Big(\frac{z}{z_\mathrm{R}}\Big)^2},
  \notag
\end{equation}
with $z_\mathrm{R}$ the Rayleigh range.
Optical signals are assumed to use an intensity-modulation and direct-detection scheme, e.g., ON-OFF keying (OOK). While exact capacity expressions for such optical intensity channels are generally intractable, useful lower bounds exist \cite{farid2007outage,lapidoth2009capacity,chaaban2022capacity}. For a given transverse offset $y$ and link distance $z$, a capacity estimation (bits/s) is given by
\begin{equation}
  C(y,z) \approx \frac{B}{2}\log_2\!\Bigg(1+\frac{\big(A\,\Phi(y,z)\,\Psi\big)^2}{2\pi e\,\sigma^2_{\mathrm{N}}}\Bigg),
  \notag
\end{equation}
where $A$ is the receiver aperture area (m$^2$), $\Psi$ is the photodetector responsivity (A/W), $B$ is the optical bandwidth (Hz), and $\sigma_{\mathrm{N}}$ is the noise current (A) \cite{maxim_spf_transimpedance}.
Under small angular jitter, the radial distance $y$ to the beam pointing direction can be approximated by $y\approx z\,\nu$, where $\nu$ is the pointing misalignment angle. Thus, the link capacity can be expressed as a function $C(z\nu,z)$. In our analysis we assume ideal acquisition and tracking (no acquisition time or tracking errors) so that performance is governed by beam misalignment and channel noise as modeled above.

Because of pointing jitter, the LISL’s capacity fluctuates with the receiver’s offset relative to the beam center. Thus, the peak data rate for the link connecting LCTs $n$ and $m$ is governed by both the jitter statistics and the inter-terminal range \cite{farid2007outage,lapidoth2009capacity,chaaban2022capacity}, i.e.,
\begin{equation}\label{eq:rate_configuration}
  \begin{aligned}
    r_{n,m} = \max_{C'} (1-\Pr\{C(z_{i_n,i_m}\nu,z_{i_n,i_m}) < C'\}) C',
  \end{aligned}
  \notag
\end{equation}
where $1-\Pr\{C(z\nu,z) < C'\}$ is the non-outage probability that the instantaneous capacity exceeds $C'$, and $r_{n,m}$ denotes the long-term effective LISL rate averaged over outages.
Over a coherent time interval, the inter-satellite range $z$ can be treated as constant, allowing the outage event to be expressed via pointing jitter as $\Pr\{C(z\nu,z) < C'\}=\Pr\{\nu>\hat{\nu}\}$ with $C'=C(z\hat{\nu},z)$.
Here, $\hat{\nu}$ is the jitter threshold beyond which the capacity falls below $C'$.
Setting an outage threshold value $\epsilon$, i.e., $\Pr\{\nu>\hat{\nu}\}\le\epsilon$, turns the optimization in \eqref{eq:rate_configuration} into choosing a misalignment threshold $\hat{\nu}$ such that the outage probability does not exceed $\epsilon$.
Given the Rayleigh model in \eqref{eq:pointing_jitter}, this yields $\hat{\nu}=\sigma_{\mathrm{J}}\sqrt{-2\ln(\epsilon)}$.
Substituting this $\hat{\nu}$, the LISL capacity between LCTs $n$ and $m$ is approximated as
\begin{equation}
  \begin{aligned}
    r_{n,m} \approx (1-\epsilon) C(z_{i_n,i_m}\sigma_{\mathrm{J}} \sqrt{-2\ln(\epsilon)},z_{i_n,i_m}), \ \forall \{n,m\} \in \mathcal{E}.
  \end{aligned}
  \notag
\end{equation}

\section*{Appendix: Proof of Lemma \ref{lemma:bounded_lagrange_multipliers}}
For given Lagrange multipliers $\lambda$, map them to new ones as $\lambda'_{i,j} = \min\{1, \lambda_{i,j}\}$, $\forall (i,j)\in\mathcal{L}$.
For $\lambda$, let the set of flows with routing cost greater than or equal to $1$ be $\mathcal{F}^+$ and those with routing cost less than $1$ as $\mathcal{F}^-$. Then, consider the flow costs for $\lambda'$. In the set $\mathcal{F}^-$, the routing cost is not changed since the routing cost was less than $1$ for $\lambda$ and is still the minimum cost for $\lambda'$.
In the set $\mathcal{F}^+$, assuming a flow's routing decisions are changed, and its cost becomes less than $1$, this new routing cost should be also the minimum cost for $\lambda$ since all links in the path have a cost less than $1$, which contradicts the fact that the minimum routing cost is greater than or equal to $1$ for flows in $\mathcal{F}^+$.
Therefore, the sets $\mathcal{F}^+$ and $\mathcal{F}^-$ remains the same for $\lambda'$ as for $\lambda$.
Since the flow rates of $\mathcal{F}^+$ computed in the dual function is always $0$ (due to the negative weights in the maximization objective) and the routing cost of flows in $\mathcal{F}^-$ remains the same, the value of the routing cost-weighted flow rate maximization is the same for $\lambda'$ and $\lambda$, i.e.,
\begin{equation}
  \begin{aligned}
      &\max_{
        \substack{
          \mathbf{q}\in \mathcal{Q}
        }
      }
      \big\{
      \sum_{(s,s')\in\mathcal{F} } q^{s,s'} (1 - \min_{
        \substack{
          \mathbf{x}^{s,s'}\in\mathcal{X}^{s,s'}
        }
      }
      \sum_{(i,j)\in\mathcal{L}}\lambda'_{i,j} x^{s,s'}_{i,j})   
      \big\}\\
      =
      &\max_{
        \substack{
          \mathbf{q}\in \mathcal{Q}
        }
      }
      \big\{
      \sum_{(s,s')\in\mathcal{F} } q^{s,s'} (1 - \min_{
        \substack{
          \mathbf{x}^{s,s'}\in\mathcal{X}^{s,s'}
        }
      }
      \sum_{(i,j)\in\mathcal{L}}\lambda_{i,j} x^{s,s'}_{i,j})   
      \big\}.
\end{aligned}
\notag
\end{equation}
Also, we have the MWM with $\lambda'$ leading to matched weights no greater than $\lambda$ since $\lambda'\leq\lambda$ as
\begin{equation}
  \begin{aligned}
&\max_{
  \substack{
    \mathbf{c}\in\mathcal{C}
  }
}
\big\{
\sum_{\{n,m\}\in\mathcal{E}} (\lambda'_{i_n,i_m} + \lambda'_{i_m,i_n})  r_{n,m} c_{n,m}
\big\}\\ 
\leq&
\max_{
  \substack{
    \mathbf{c}\in\mathcal{C}
  }
}
\big\{
\sum_{\{n,m\}\in\mathcal{E}} (\lambda_{i_n,i_m} + \lambda_{i_m,i_n})  r_{n,m} c_{n,m}
\big\}.
  \end{aligned}
  \notag
\end{equation}
Due to the above facts, $g(\lambda')\geq g(\lambda)$. Since we can always find such $\lambda'$ for any $\lambda$ as well as for any $\lambda^*$, there exist optimal Lagrange multipliers $\lambda^*$ with all elements less than or equal to $1$. Thus, the difference between $\lambda^{[1]}=\mathbf{0}$ and $\lambda^*$ can be bounded by $\|\mathbf{1}^{|\mathcal{E}|\times1}\|$. This completes the proof.

\section*{Proof of Theorem \ref{theorem:convergence:gnn_training}}

\begin{proof}
For convenience, we denote 
\begin{equation}
  \begin{aligned}
    \zeta(\mathbf{w}^{[k]}) \triangleq - \sum_{(i,j)} \delta(\lambda)_{i,j} \nabla_{\mathbf{w}^{[k]}}\lambda_{i,j}|_{\lambda=\mu(\mathbf{S}^{[k]}, \mathbf{R}^{[k]} | \mathbf{w}^{[k]})}.
  \end{aligned}
\end{equation}
The smoothness of $L(\mathbf{w})$ implies that $\forall \mathbf{w}', \mathbf{w}$,
\begin{equation}
  \begin{aligned}
    L(\mathbf{w}') \leq L(\mathbf{w}) + \nabla_\mathbf{w} L(\mathbf{w})^{\rm T} (\mathbf{w}' - \mathbf{w}) + \frac{\gamma}{2}\|\mathbf{w}' - \mathbf{w}\|^2.
  \end{aligned}
  \notag
\end{equation}
Substituting $\mathbf{w}^{[k+1]}$, $\mathbf{w}^{[k]}$ and \eqref{eq:gnn_parameter_update}, i.e., $\mathbf{w}^{[k+1]} = \mathbf{w}^{[k]} + \alpha^{[k]} \zeta(\mathbf{w}^{[k]})$, in the smoothness expression, we have
\begin{equation}
  \begin{aligned}
&\quad L(\mathbf{w}^{[k+1]}) - L(\mathbf{w}^{[k]}) \\
&\leq -\alpha^{[k]}  \nabla_{\mathbf{w}^{[k]}} L(\mathbf{w}^{[k]})^{\rm T} ( \zeta(\mathbf{w}^{[k]}) ) + \frac{\gamma}{2} (\alpha^{[k]})^2 \|\zeta(\mathbf{w}^{[k]})\|^2.
  \end{aligned}
  \notag
\end{equation}
Taking the expectation on both sides, we have
\begin{equation}
  \begin{aligned}
    \!&\quad\ \expt[L(\mathbf{w}^{[k+1]})] - \expt[L(\mathbf{w}^{[k]})]\\
    \!&\leq\! -\alpha^{[k]}  \expt[\nabla_{\mathbf{w}^{[k]}} L(\mathbf{w}^{[k]})^{\rm T}\zeta(\mathbf{w}^{[k]})] + \frac{\gamma}{2} (\alpha^{[k]})^2 \expt[\|\zeta(\mathbf{w}^{[k]})\|^2]\\
    \!&\leq\! -\alpha^{[k]}\!\|\nabla_{\mathbf{w}^{[k]}} L(\mathbf{w}^{[k]})\|^2\!+\!\frac{\gamma}{2} (\alpha^{[k]})^2 (\|\nabla_{\mathbf{w}^{[k]}} L(\mathbf{w}^{[k]})\|^2\! +\!\sigma_{L}^2).
  \end{aligned}
  \notag
\end{equation}
Since $\alpha^{[k]}$ is decaying, we can always find a $K'$ such that $\alpha^{[k]}\leq \frac{1}{\gamma}$ for all $k\geq K'$. Considering $k>K'$, apply telescoping sum on both sides from $K'$ to $k$, we have
\begin{equation}
  \begin{aligned}
    &\quad\ \sum_{i=K'}^{k} (\alpha^{[i]} -\frac{\gamma}{2} (\alpha^{[i]})^2) \|\nabla_{\mathbf{w}^{[k]}} L(\mathbf{w}^{[k]})\|^2 \\
    &\leq \expt[L(\mathbf{w}^{[K']})] - \expt[L(\mathbf{w}^{[k+1]})] + \frac{\gamma\sigma_{L}^2}{2} \sum_{i=K'}^{k} (\alpha^{[i]})^2.
  \end{aligned}
  \notag
\end{equation}
As $\alpha^{[i]}\leq \frac{1}{\gamma}$, $\alpha^{[i]} -\frac{\gamma}{2} (\alpha^{[i]})^2\geq\alpha^{[i]} -\frac{\gamma}{2}\cdot \frac{1}{\gamma} \cdot \alpha^{[i]} = \frac{1}{2}\alpha^{[i]}$, for all $i\geq K'$.
Moreover, note that the dual function $g(\mu(\mathbf{S}, \mathbf{R} | \mathbf{w}))$ is bounded because the GNN function $\mu(\cdot | \mathbf{w})$ is bounded between $0$ and $1$, and the traffic rates as well as link capacities are also bounded. Consequently, the loss function $L(\mathbf{w})$ is also bounded.
Let $\hat{L}$ and $\check{L}$ be the upper and the lower bound of the loss function $L(\mathbf{w})$, i.e., $\check{L} = L(\mathbf{w})$ and $\hat{L}= L(\mathbf{w})$ $\forall \mathbf{w}$. Then, their difference is bounded and finite as $\Delta_L = \hat{L} - \check{L}$.
Therefore, we have
\begin{equation}
  \begin{aligned}
  \frac{1}{2}\!\sum_{i=K'}^{k} \alpha^{[i]} \|\nabla_{\mathbf{w}^{[k]}} L(\mathbf{w}^{[k]})\|^2\leq \hat{L} - \check{L}+ \frac{\gamma\sigma_{L}^2}{2} \sum_{i=K'}^{k} (\alpha^{[i]})^2.
  \end{aligned}
  \notag
\end{equation}
Furthermore, we find the minimum of $\|\bar{\zeta}(\mathbf{w}^{[i]})\|^2$ over $i=K',\dots,k$ on the left-hand side, i.e.,
\begin{equation}
  \begin{aligned}
        &\frac{1}{2} \lim_{k\to\infty} \min \{\|\nabla_{\mathbf{w}^{[i]}} L(\mathbf{w}^{[i]})\|^2\}_{i=K'}^{k}\!\!\sum_{K'}^{k} \alpha^{[i]} \\
    \leq& \frac{1}{2}\!\! \sum_{i=K'}^{k}\! \alpha^{[i]} \|\nabla_{\mathbf{w}^{[i]}} L(\mathbf{w}^{[i]})\|^2\!\leq\! \hat{L}\!-\!\check{L}\!+\!\frac{\gamma\sigma_L^2}{2}\!\sum_{i=K'}^{k} (\alpha^{[i]})^2.
  \end{aligned}
  \notag
\end{equation}
Rearrange the inequality, we have
\begin{equation}
  \begin{aligned}
    \lim_{k\to\infty}\!\!\min \{\|\nabla_{\mathbf{w}^{[i]}} L(\mathbf{w}^{[i]})\|^2\}_{i=K'}^{k}\!\!\leq\!\!\frac{2\Delta_L \!+\! \gamma\sigma_L^2\!\sum_{i=K'}^{k} (\alpha^{[i]})^2}{\sum_{i=K'}^{k} \alpha^{[i]}}.
  \end{aligned}
  \notag
\end{equation}
As $\Delta_L$, $\gamma$ and $\sigma_L$ are all finite, the convergence of the p-series, $\sum_{i=K'}^{k}\alpha^{[i]}$ and $\sum_{i=K'}^{k}(\alpha^{[i]})^2$ as $k\to\infty$ when $0.5\leq \beta < 1$ \cite{boyd2003subgradient}, implies $\lim_{k\to\infty}\min \{\|\nabla_{\mathbf{w}^{[i]}} L(\mathbf{w}^{[i]})\|^2\}_{i=K'}^{k} \leq \mathcal{O}(k^{-(1-\beta)})$. This also implies that $\lim_{k\to\infty}\min \{\|\nabla_{\mathbf{w}^{[i]}} L(\mathbf{w}^{[i]})\|\}_{i=1}^{k} \leq \mathcal{O}(k^{-(1-\beta)})$, which completes the proof.
\end{proof}

\bibliography{main}
\bibliographystyle{IEEEtran}

\end{document}